\documentclass[12pt,reqno,a4paper]{amsart}

\usepackage[margin=3cm,footskip=1cm]{geometry}

\usepackage{amsmath}
\usepackage{amsfonts}
\usepackage{amsthm}
\usepackage{amstext}
\usepackage{enumerate}
\usepackage[latin1]{inputenc}
\usepackage[shortlabels]{enumitem}

\usepackage{amssymb}
\usepackage{latexsym}

\usepackage{graphicx}
\usepackage{verbatim}

     \newcommand{\dist}{{\operatorname{dist}}}

     \newcommand{\Ran}{{\operatorname{Ran}}}

     \newcommand{\N}{{\mathbb{N}}}

     \newcommand{\R}{{\mathbb{R}}}

\newcommand{\slim}{{\rm s-}\lim}
\newcommand{\s}{{\rm s}}
\newcommand{\e}{{\rm e}}

\renewcommand{\i}{{\rm i}}
\renewcommand{\d}{{\rm d}}
\newcommand{\D}{{\rm D}}

\newcommand{\aux}{{\rm aux}}

\renewcommand{\Re}{{\rm Re}\,}
\renewcommand{\Im}{{\rm Im}\,}

\newcommand\inp[2][]{#1 \langle #2#1\rangle}
\newcommand\parb[2][]{#1 \big ( #2#1\big )}
\newcommand\parbb[2][]{#1 \Big ( #2#1\Big )}

\newcommand{\grad}{{\rm grad }\,}

\renewcommand{\exp}{{\rm exp}}
\newcommand{\mand}{\text{ and }}
\newcommand{\mfor}{\text{ for }}

\newcommand{\mForall}{\text{For all }}

\newcommand{\mon}{\text{ on }}

\newcommand{\vB}{{\mathcal B}}

\newcommand{\vD}{{\mathcal D}}

\newcommand{\vG}{{\mathcal G}}

\newcommand{\vH}{{\mathcal H}}

\newcommand{\vM}{{\mathcal M}}

\newcommand{\vO}{{\mathcal O}}

     \theoremstyle{plain}
     \newtheorem{thm}{Theorem}[section]
     \newtheorem{proposition}[thm]{Proposition}
     \newtheorem{lemma}[thm]{Lemma}
      \newtheorem{corollary}[thm]{Corollary}
     \theoremstyle{definition}

     \newtheorem{cond}[thm]{Condition}

     \newtheorem{remark}[thm]{Remark}
     \newtheorem{remarks}[thm]{Remarks}

\newtheorem*{remarks*}{Remarks}
\newtheorem*{remark*}{Remark}

     \numberwithin{equation}{section}

\title[Scattering theory for Riemannian Laplacians]{Scattering theory for Riemannian Laplacians}

\thanks{
This work was essentially done during K.I.'s stay in Aarhus University (academic
year 2009-2010).
He would like to express his gratitude for financial support from  FNU 160377
(2009--2011) as
well as  from  JSPS Wakate (B) 21740090
(2009--2012). E.S. thanks H.D. Cornean and I. Herbst for many preliminary
discussions of scattering theory on manifolds \cite{CHS2}. 
We thank H. Kumura for bringing our attention on his  work \cite{K2}.}

\author{K. Ito}
\address[K. Ito ]{Graduate School of Pure and Applied Sciences, University of Tsukuba\\
1-1-1 Tennodai, Tsukuba Ibaraki, 305-8571 Japan.  
} 
\email{ito-ken@math.tsukuba.ac.jp}

\author{E. Skibsted}
\address[E. Skibsted]{Institut for  Matematiske
Fag \\
Aarhus Universitet\\ Ny Munkegade  8000 Aarhus C,
Denmark}
\email{skibsted@imf.au.dk}

\begin{document}
\begin{abstract}
In this paper we introduce a notion of scattering theory for the
Laplace-Beltrami operator on  non-compact,
connected  and complete  Riemannian manifolds. A principal  condition is  given by
a certain positive lower bound of the second fundamental form of
angular submanifolds at infinity. Another condition is  certain 
bounds of derivatives up to order one  of the trace of this quantity. These
conditions are shown to be optimal for existence and completeness of a
 wave
operator. Our theory does not involve prescribed asymptotic behaviour  of the
metric at infinity (like  asymptotic Euclidean or hyperbolic
metrics studied  previously in the literature). A consequence of the
theory is spectral theory for the
Laplace-Beltrami operator including identification of the continuous spectrum and absence of
singular continuous spectrum.
\end{abstract}
\maketitle
\tableofcontents

\section{Introduction and results}\label{sec:introduction}
In this paper we introduce a notion of
scattering theory for the
Laplace-Beltrami operator on a rather general type of  non-compact manifold. In particular
we do not  impose asymptotics of the metric at infinity. Immediate consequences
include  identification of the continuous spectrum and absence of
singular continuous spectrum. We also show that our wave operator
implements a certain family of commuting asymptotic observables. To our
knowledge  most previous works on spectral and
scattering theory for the
Laplace-Beltrami operator on manifolds require asymptotics of the
metric  at infinity (or at least asymptotics of the large ball volume), see for example \cite {Bo, Do, FH, IN, K1, K2, Me,
  MZ}. Among these   works probably \cite {Do} 
is the closest to our setup. In fact Donnelly's assumptions include
the existence of a certain {\it exhaustion function} $b$ resembling the
function $r$ appearing in our assumptions, see Conditions
\ref{cond:conv} and \ref{cond:10.9.2.11.13} below. However he needs
asymptotics of the Hessian of $b^2$ while our Condition 
\ref{cond:conv} is a lower bound only of the Hessian of the analogous 
function $r^2$. Moreover the main issue of our paper is scattering 
theory while \cite {Do} only deals with spectral theory. In a companion
paper \cite {IS} we prove  absence of embedded  eigenvalues under
weaker conditions than considered in the present paper. All of our
results generalize to Schr\"odinger operators on manifolds (with
short-range potentials). We state and prove our results in this more
general context.

The comparison dynamics used to define our wave operator is
constructed from a certain family of geodesics for the (full) metric
in the spirit of primarily  \cite {HS, CHS1}. In this sense it is
non-perturbatively constructed. Nevertheless it provides a simple
explicit description of the large time behaviour of continuous
spectrum wave packets which is a fundamental  goal of
scattering theory  \cite {DG1}.

Let $(M,g)$ be a connected  complete $d$-dimensional Riemannian manifold, $d\ge 2$.
In the present paper we 
discuss the scattering theory for the Schr\"odinger operator
\begin{align}
H=-\tfrac12\triangle +V=H_0+V
\label{eq:10.5.4.3.24}
\end{align}
on the Hilbert space ${\mathcal H}=L^2(M)=L^2(M,(\det g)^{1/2}\mathrm{d}x)$.
Here $\triangle$ is the Laplace-Beltrami operator:
In any local coordinates $x$, if $g=g_{ij}\mathrm{d}x^i\otimes \mathrm{d}x^j$, then
\begin{align*}
-\triangle=p_i^*g^{ij}p_j=\frac{1}{(\det g)^{1/2}}p_i (\det g)^{1/2}g^{ij}p_j,&&
p_i=-\mathrm{i}\partial_{i}.
\end{align*} 
Note that indeed $p_i^*=(\det g)^{-1/2}p_i (\det g)^{1/2}$ is the adjoint of $p_i$.
Since our conditions will include that the potential $V=V(x)$ is bounded (Condition
\ref{cond:10.6.1.16.24} given below states that it is bounded and
short-range) $H=H_0+V$ is essentially self-adjoint on $C_{\rm
  c}^\infty(M)$. Concerning geometric notions
appearing below 
we refer to \cite{Chavel} (see also \cite{Jost} or \cite{Mi}).

We first impose, cf.\ \cite{K1},
\begin{cond}\label{cond:diffeo} 
There exists a relatively compact open set $O\Subset M$ such that 
the boundary $\partial O$ is smooth and the exponential map restricted to outward normal vectors:
$\exp \colon N^+\partial O\to E:=M\setminus \overline{O}$ is diffeomorphic.
\end{cond}   
Then we call a component of $E$ an \textit{end} and such $M$ a \textit{manifold with ends}.

The distance function $r(x)=\dist(x,\partial O)$, $x\in E$, belongs to
$C^\infty(E)$. In a neighbourhood of $\partial
O$, say $\vO$, we have an extension of $r$, say $\tilde r$, with the
property $\tilde r\in C^\infty(\vO)$.  In fact we can choose this
extension as the ``signed distance function''. We can then  construct
an extension  $\tilde r\in C^\infty(M)$ for which we  may assume
(although these requirements are  not essential) $-1\le
\tilde r<0$ and $|\nabla \tilde r|\le 1$ on $O$ (this is by considering a
certain composition of functions).
In the following we use the notation $r$ for this extended
function. We point out that our main results will  be independent of the
extension procedure, however we prefer in proofs to work entirely with objects
defined on the whole of $M$ (rather than with some defined on $E$ only).

We denote the Levi-Civita connection on $TM$ by 
$\nabla\colon \Gamma(TM)\to\Gamma(TM\otimes T^*M)$.
In general, the connection $\nabla$ extends naturally to 
\begin{align*}
\nabla\colon \Gamma((TM)^{\otimes p}\otimes (T^*M)^{\otimes q})
\to\Gamma((TM)^{\otimes p}\otimes (T^*M)^{\otimes (q+1)})
\end{align*}
in the following way (cf. \cite[p. 31]{Chavel}): For any
\begin{align*}
t=(t^{i_1\dots i_p}_{j_1\dots j_q})
\in \Gamma((TM)^{\otimes p}\otimes (T^*M)^{\otimes q})
\end{align*}
$\nabla t\in \Gamma((TM)^{\otimes p}\otimes (T^*M)^{\otimes (q+1)})$
is given by 
\begin{align*}
(\nabla t)^{i_1\cdots i_p}_{j_0j_1\cdots j_q}
 =\partial_{j_{0}} t^{i_1\cdots i_p}_{j_1\cdots j_{q}}
+\sum_{s=1}^p \Gamma^{i_s}_{j_0 k}t^{i_1\cdots i_{s-1}ki_{s+1}\cdots  i_p}_{j_1\cdots j_{q}}
-\sum_{s=1}^q \Gamma^{k}_{j_0 j_s }t^{i_1\cdots i_p}_{j_1\cdots j_{s-1}kj_{s+1}\cdots j_{q}},
\end{align*}
where $\Gamma^k_{ij}=2^{-1}g^{kl}(\partial_i g_{lj}+\partial_j g_{li}-\partial_lg_{ij})$ are the Christoffel symbols.
For example, it follows that $\nabla r^2=\mathrm{d}r^2\in \Gamma(T^*M)$ and 
$\nabla^2r^2=\nabla\nabla r^2\in \Gamma(T^*M\otimes T^*M)$.
The operator $\nabla^2$ gives the geometric Hessian, and in local coordinates
\begin{align}
(\nabla^2 r^2)_{ij}
=\partial_i\partial_j r^2-\Gamma^k_{ij}\partial_kr^2.
\label{eq:10.12.20.14.56}
\end{align}
We note that $(\nabla^2 r^2)_{ij}$ are the coefficients of the
principal part of a Mourre type commutator,
 cf.\ Corollary~\ref{cor:10.10.13.15.00} and
Lemma~\ref{lemma:10.1.23.9.52} (for the analogous  statement in
Classical  Mechanics  see the end
of Subsection \ref{Classical mechanics considerations}).
\begin{cond}[Mourre type condition]\label{cond:conv}
There exist  $\delta\in (0,1]$ and $r_0\geq 0$ such that 
\begin{align}
\nabla^2 r^2 \ge (1+\delta)g\mfor 
r\geq r_0.\label{eq:10.2.8.3.55}
\end{align}
\end{cond}

Note that (\ref{eq:10.2.8.3.55}) is an inequality of quadratic forms on fibers of $TM$.
The condition (\ref{eq:10.2.8.3.55}) can also be formulated in terms of the second fundamental form of the 
\textit{angular manifolds}
$S_r=\{x\in E;\ r(x)=r\}\cong \partial O$.
We let $\iota_r\colon S_r\hookrightarrow M$ be the inclusion and 
$\mathrm{D}=\iota_r^* \circ\nabla$ (cf. \cite[Proposition 2.3]{Chavel}).
Then (\ref{eq:10.2.8.3.55}) is equivalent to the following inequality in the sense of quadratic forms on $TS_r$:
\begin{equation}\label{eq:1} 
\mathrm{D}\nabla r\ge \frac{(1+\delta)}{2r}\iota_r^* g\mfor 
r> r_0.
\end{equation} 
In fact, a computation in the geodesic spherical coordinates shows that 
\begin{align}\label{eq:17}
\nabla^2 r^2=2\mathrm{d}r\otimes\mathrm{d}r \oplus 2r\mathrm{D}\nabla r,&&
g=\mathrm{d}r\otimes\mathrm{d}r\oplus \iota_r^* g.
\end{align} Here the direct sum decompositions correspond to the
orthogonal splitting $TM_x\cong (NS_r)_x\oplus (TS_r)_x$ at any point $x\in
S_r$.

As we can see from (\ref{eq:10.12.20.14.56}) the inequality
(\ref{eq:10.2.8.3.55}) is a condition on derivatives  of the metric
tensor $g$ up to first order (as well as  on derivatives of the function $r^2$ of course).
The  condition \eqref{eq:10.9.2.23.19} below is on  derivatives up to
second  order.

\begin{subequations}
\begin{cond}[Quantum Mechanics bound]\label{cond:10.9.2.11.13} 
There exists $\kappa\in (0,1/2)$ such that 
\begin{equation}
  |\d \triangle r^2|\leq C \langle r\rangle^{-1/2-\kappa}.
\label{eq:10.9.2.23.19}
\end{equation}
\end{cond}

We used the standard notation $\langle r\rangle =(1+r^2)^{1/2}$. Due
to \eqref{eq:10} given below we have $\partial_i \triangle
r^2=g^{jk}(\nabla^3 r^2)_{ijk}$.
We notice
that it is a consequence of Condition~\ref{cond:10.9.2.11.13} that
\begin{equation}
  \label{eq:49}
  \triangle r=O( r^{-1/2-\kappa})\mand |\d  \triangle r|=O( r^{-3/2-\kappa}),
\end{equation} in fact \eqref{eq:10.9.2.23.19} and \eqref{eq:49} are 
equivalent  for any  $\kappa\in (0,1/2)$. Whence yet another
equivalent condition (given in terms of the mean curvature) is 
\begin{equation}
  \label{eq:49m}
  \mathop{\mathrm{tr}}{}(\mathrm{D}\nabla r)=O( r^{-1/2-\kappa})\mand |\d  \mathop{\mathrm{tr}}{}(\mathrm{D}\nabla r)|=O( r^{-3/2-\kappa}).
\end{equation} 
 \end{subequations} 
 The first bound of \eqref{eq:49m} implies an upper bound of the ball volume
 growth of the form $\exp\parb{Cr^{1/2-\kappa}}$, $C>0$,  and in general no
 better. Similarly Condition \ref{cond:conv} implies the 
 power type   
 lower bound of the ball volume
 growth $c r^\sigma$ with  $\sigma=(\delta+1)(d-1)/2+1$ and $c>0$.

In the analysis of the Classical Mechanics in Section~\ref{Classical mechanics considerations}
we do not need 
Condition~\ref{cond:10.9.2.11.13}.


Finally we impose a short-range condition on $V$:
\begin{cond}\label{cond:10.6.1.16.24}
The potential $V\in L^\infty (M;\mathbb{R})$ satisfies for some $\eta\in (0,1]$
\begin{equation}\label{eq:60}
|V(x)|\le C\langle r\rangle^{-1-\eta}.
\end{equation}
\end{cond}   

Under the above setting we prove the existence and the completeness of the wave operator.
Define the free propagator $U(t)$, $t>0$, by
\begin{subequations}
\begin{align}
U(t)&=\mathrm{e}^{\mathrm{i}K(t,\cdot)}\mathrm{e}^{-\mathrm{i}\frac{\ln t}{2}A};\label{eq:10.5.4.22.21}\\
K(t,x)&=\tfrac{r(x)^2}{2t}, \label{eq:59}\\
A&=\i[H_0,r^2]=\tfrac{1}{2}\{(\partial_i r^2)g^{ij}p_j+p_i^* g^{ij}(\partial_j r^2)\}.\label{eq:6}
\end{align}
  \end{subequations}
Here $\mathrm{e}^{-\mathrm{i}\frac{\ln t}{2}A}$ is 
called a dilation with respect to $r$.
If we define a  flow $\omega=\omega(t,x)$, $(t,x)\in (0,\infty)\times M$, by
\begin{align}
\partial_t \omega^i=-\tfrac{1}{2t}g^{ij}(\omega)(\partial_jr^2)(\omega),&&
\omega(1,x)=x,
\label{09.12.17.22.23}
\end{align}
then for $u\in \vH$
\begin{align}\label{eq:61}
\mathrm{e}^{-\mathrm{i}\frac{\ln t}{2}A}u(x)=J(\omega(t,x))^{1/2}
\Bigl(\frac{\det g(\omega(t,x))}{\det g(x)}\Bigr)^{1/4}u(\omega(t,x)),
\end{align}
where $J$ is the relevant Jacobian.
In fact,
using (\ref{09.12.17.22.23}) and the relation
\begin{equation}\label{eq:10}
\triangle u =\mathop{\mathrm{tr}}{}(\nabla^2 u)=g^{ij}(\nabla^2 u)_{ij},
\end{equation}
we can show 
\begin{align*}
J(\omega(t,x))^{1/2}
\Bigl(\frac{\det g(\omega(t,x))}{\det g(x)}\Bigr)^{1/4}
=\exp\Bigl(\int^t_1\tfrac{1}{4s}(-\triangle r^2)(\omega(s,x))\,\mathrm{d}s\Bigr).
\end{align*} The right hand side of this identity is a geometric 
invariant, and indeed it shows in  combination with the group property
$\omega(t,\omega(s,x))=\omega(ts,x)$  the formula
\eqref{eq:61}. Note, as a consequence of \eqref{eq:61}, that
$C^\infty_{\rm c}(M)$ is left invariant under  dilations (in particular
the generator $A$ is essentially self-adjoint on $C_{\rm
  c}^\infty(M)$).
We also note that $\omega$ fixes $\partial O$ and, moreover,
\begin{align}
\omega(t,x)=\exp|_{N^+\partial O}\bigl[\tfrac{1}{t}(\exp|_{N^+\partial O})^{-1}(x)\bigr]
\mfor (t,x)\in (0,\infty)\times E.
\label{11.6.6.8.3}
\end{align}
Hence $\mathrm{e}^{-\mathrm{i}\frac{\ln t}{2}A}$ is unitary on $\vH_\aux:=L^2(E)\subset \vH$
and $(\vH_\aux)^{\perp}=L^2(O)\subset \vH$, respectively. 
By (\ref{11.6.6.8.3}) 
$\mathrm{e}^{-\mathrm{i}\frac{\ln t}{2}A}|_{\vH_\aux}$ is the ``geodesic
dilation'' on $E$ (since the composition part is given in geodesic spherical coordinates by $r\to r/t$),
while $\mathrm{e}^{-\mathrm{i}\frac{\ln t}{2}A}|_{(\vH_\aux)^{\perp}}$
does not have a similar geometric meaning.
Moreover, due to the eikonal equation 
\begin{equation}\label{eq:3}
|\mathop{\mathrm{grad}} r|^2_g=g^{ij}(\partial_ir)(\partial_jr)=1\mon  E
\end{equation}
it follows that $K$ is
a solution to the Hamilton-Jacobi  equation
\begin{align}
\partial_tK=-\tfrac{1}{2}g^{ij}(\partial_i K)(\partial_j K)\mon E.\label{eq:9.12.19.054}
\end{align}

\begin{thm}\label{thm:10.5.31.23.23}
Let $(M,g)$ be a connected complete Riemannian manifold satisfying Conditions \ref{cond:diffeo}--\ref{cond:10.9.2.11.13}, 
and $V$ a potential satisfying Condition~\ref{cond:10.6.1.16.24}.
Then, for the Schr\"odinger propagator $\mathrm{e}^{-\mathrm{i}tH}$ for (\ref{eq:10.5.4.3.24}) and the free propagator (\ref{eq:10.5.4.22.21})
there exists the wave operator
\begin{align*}
\Omega_+=\slim_{t\to+\infty} \e^{\i tH}U(t)P_{\mathrm{aux}},
\end{align*}
where $P_{\mathrm{aux}}$ is the projection onto ${\mathcal H}_{\mathrm{aux}}$.
Moreover 
there exists the limit
\begin{align*}
\widetilde\Omega_+=\slim_{t\to+ \infty}U(t)^*\e^{-\mathrm{i}tH}P_{\mathrm{c}},
\end{align*}
where $P_{\mathrm{c}}$ is the projection onto the continuous subspace
${\mathcal H}_{\mathrm{c}}(H)=\chi_{(0,\infty)}(H)\vH
$ for $H$.

Finally $\Omega_+$ is complete, i.e. 
\begin{equation}
  \label{eq:7}
  \widetilde\Omega_+=\Omega_+^*,\;\;
\Omega_+^*\Omega_+=P_{\aux} \mand \Omega_+\Omega_+^*=P_{\mathrm{c}}.
\end{equation}

\end{thm}

 Here we used the notation $\chi_\vO$ to denote the characteristic
 function of a  subset $\vO\subseteq \R$. 
 Note that $U(t)P_{\aux}$ and $\Omega_+$ are independent of the extension of $r$ to $O$.
 The fact that $H$ does not have positive eigenvalues is proved under
 weaker conditions in \cite{IS}  and will not be discussed in this
 paper. 
 It follows  by a standard local
 compactness argument that the negative spectrum of $H$ (if not empty)
 consists
 of finite multiplicity eigenvalues  accumulating at most at zero.

Note that  $t\cdot r(\omega(t,x))=r(x)$ for $(t,x)\in (0,\infty)\times E$. By this formula it follows
readily that 
\begin{equation}\label{eq:2}
  \Omega_+^*H\Omega_+=M_fP_{\mathrm{aux}};\;f=2^{-1} r(\cdot)^2.
\end{equation} Here $M_f$ means the operator given by multiplication
by $f$ (defined maximally on $\vH$).
 Consequently we immediately deduce 
\begin{corollary}[Spectrum]\label{cor:asymwb} The continuous spectrum
$\sigma_{\rm c}(H)=\sigma(H_{\rm c})=[0,\infty)$ and the
  singular continuous spectrum of $H$ is absent (i.e. 
$\sigma_{\mathrm{sc}}(H)=\emptyset$).
\end{corollary}
Note that under Conditions \ref{cond:diffeo} and
\ref{cond:10.9.2.11.13} the essential spectrum $\sigma_{\rm ess}
(H_0)=[0,\infty)$, see \cite [Theorem 1.2]{K1}. On the other hand the second  
part of the corollary on 
the  singular continuous spectrum of $H$ is new.

 As another  corollary, the existence of ``the asymptotic speed'' follows (see for example \cite{DG1} for notation).
\begin{corollary}[Asymptotic observables]\label{cor:asymw}
In the space ${\mathcal H}_{\mathrm{c}}(H)$ there exists the $*$-representation
\begin{equation*}
   \omega_{\infty}^+  =\s -C_{\mathrm{c}}(M)-\lim_{t\to
   +\infty}\e^{\i tH}\omega(t,\cdot )\e^{-\i tH}.
\end{equation*} In particular the asymptotic speed
\begin{equation*}
   r(\omega_{\infty}^+) = \s -C_{\mathrm{c}}(\R)-\lim_{t\to
   +\infty}\e^{\i tH}\;\tfrac{r(\cdot )}{t}\;\e^{-\i tH}
\end{equation*} exists as a self-adjoint operator on $\vH_{\mathrm{c}}(H)$. This
operator is positive   with zero kernel.

Moreover, for all $\phi\in C_{\mathrm{c}}(M)$
\begin{equation}\label{eq:2bb}
  \phi(\omega_{\infty}^+)  =\Omega_+M_{\phi}\Omega_+^*\mand H_{\mathrm{c}}=2^{-1}r(\omega_{\infty}^+ )^2.
\end{equation}
\end{corollary}

\begin{remarks}\label{remarks:intro}
\begin{enumerate}[\normalfont 1)]
\item\label{item:2}
 A principal  virtue of the formula (\ref{eq:10.5.4.22.21}) is
its intrinsic ``position space''
nature. In Euclidean scattering such formula appeared first in \cite
{Y}; for later developments see \cite {DG2, HS,  CHS1}. It was
conceived in \cite { CHS2} for  a general geometric
setting.  Another virtue of
(\ref{eq:10.5.4.22.21}) is that time reversal invariance applies yielding
existence and completeness of a similar   wave operator $\Omega_-$
(constructed by taking $t\to -\infty$). Whence  Theorem
\ref{thm:10.5.31.23.23} defines a scattering theory that includes  a unitary
scattering operator, however we shall not elaborate  here.

\item \label{item:1}We note that Conditions \ref{cond:conv}--\ref{cond:10.6.1.16.24} in some sense are optimal, see Subsection
\ref{Generalizations} for  counter examples to Theorem
\ref{thm:10.5.31.23.23} under the  slight relaxation of conditions 
given by allowing either $\delta=0$ in \eqref{eq:10.2.8.3.55}, $\kappa=0$ in 
\eqref{eq:10.9.2.23.19}  or $\eta=0$ in \eqref{eq:60}.

\item \label{item:3}
If we denote the \textit{time-dependent generator} of $U(t)$ by $G(t)$ 
then in the geodesic spherical coordinates
\begin{align*}
G(t)=\tfrac{1}{2}p_r^*p_r-\tfrac{1}{2}\Bigl(p_r-\frac{r}{t}\Bigr)^*\Bigl(p_r-\frac{r}{t}\Bigr)\mon E;\;p_r:=(\partial_kr)g^{kl}p_l.
\end{align*}
By arguments motivated by classical mechanics the second term is
\textit{short-range}. In fact  we also  have that
$G(t)=H_0-W(t)-\alpha(t)$ where 
\begin{align*}
  W(t):=\tfrac{1}{2}(p_i-\partial_i K)^*g^{ij}(p_j-\partial_j K)\mand
  \alpha(t):=(\partial_tK)+\tfrac{1}{2}g^{ij}(\partial_iK)(\partial_jK),
\end{align*} and we prove in  this paper, more generally,  that $W(t)$ 
    is
short-range. This is in fact the heart of the proof of Theorem
\ref{thm:10.5.31.23.23}.  
Whence the generator of $U(t)$  differs 
from the  one-dimensional radial Laplacian by a  short-range term, see
 \cite{IN} for a similar relationship.
   
\item \label{item:4} The subset $O$ is not uniquely determined in Condition~\ref{cond:diffeo}, 
but the wave operator is nevertheless (at least
partially) in some sense 
unique. 
We will discuss this issue  in
Subsection~\ref{subsec:Uniqueness}. We show that there is
an explicit dependence on the  one-parameter family of sets $O_a\supseteq
O$ induced
by the outward geodesic flow ($\partial O_a=\{r=a\}$; $a\geq 0$). This idea is exploited in Subsection
\ref{sec:manifold-with-pole} where we introduce a stronger condition
than Condition~\ref{cond:diffeo} (regularity of the geodesic flow from
a point rather than from a submanifold). Our main results are easily
implemented in this setting, although it is also possible (as an
alternative way of showing  results in this setting) to mimic the
procedure in the bulk of the paper.

\end{enumerate}
\end{remarks}

\subsection{Uniqueness of the wave operator}\label{subsec:Uniqueness}

Let us assume Conditions~\ref{cond:diffeo}--\ref{cond:10.6.1.16.24}.
We set for $a\ge 0$
\begin{align*}
O_a=\{x\in E;\, r(x)<a\}\cup \overline{O},&&r_a(x)=r(x)-a,
\end{align*} 
and decorate various  quantities defined previously 
with respect to $O_a$ and $r_a$ by the subscript $a$.
In particular we discuss the strong limits
\begin{align*}
V_a&=V_a^0=\slim_{t\to+\infty}U(t)^*U_a(t)P_{\mathrm{aux}}^{(a)}, \\
V^a&=V_0^a=\slim_{t\to+\infty}U_a(t)^*U(t)P_{\mathrm{aux}}.
\end{align*}

For $u\in {\mathcal H}$
\begin{align*}
U_a(t)u(x)
=\mathrm{e}^{\mathrm{i}r_a(x)^2/2t}\Bigl(\det\frac{\partial \omega_a(t,x)}{\partial x}\Bigr)^{1/2}
\Bigl(\frac{\det g(\omega_a(t,x))}{\det g(x)}\Bigr)^{1/4}u(\omega_a(t,x))
\end{align*}
and since $\omega(t,\cdot)^{-1}(x)=\omega(1/t,x)$
\begin{align*}
U_a(t)^*u(x)
=\mathrm{e}^{-\mathrm{i}tr_a(x)^2/2}\Bigl(\det\frac{\partial \omega_a(1/t,x)}{\partial x}\Bigr)^{1/2}
\Bigl(\frac{\det g(\omega_a(1/t,x))}{\det g(x)}\Bigr)^{1/4}
u(\omega_a(1/t,x)).
\end{align*}

We note that the flow $\omega_a$ satisfies 
\begin{align}
r_a(\omega_a(t,x))=r_a(x)/t\mbox{ for } x\in E_a.
\label{eq:11.5.23.7.35}
\end{align} 

Let $u\in {\mathcal H}$.
Then for $x\in O_{a/t}$
\begin{align*}
U(t)^*U_a(t)P_{\mathrm{aux}}^{(a)}u(x)=0,
\end{align*}
and for $x\in E_{a/t}$
\begin{align*}
&U(t)^*U_a(t)P_{\mathrm{aux}}^{(a)}u(x)\\
&=\mathrm{e}^{\mathrm{i}r_a(\omega(1/t,x))^2/2t-\mathrm{i}tr(x)^2/2}
\Bigl(\det\frac{\partial \omega_a(t,\cdot)}{\partial x}(\omega(1/t,x))\Bigr)^{1/2}
\Bigl(\det\frac{\partial \omega(1/t,x)}{\partial x}\Bigr)^{1/2}\\
&\phantom{{}={}}\Bigl(\frac{\det g(\omega_a(t,\omega(1/t,x)))}{\det g(\omega(1/t,x))}\Bigr)^{1/4}
\Bigl(\frac{\det g(\omega(1/t,x))}{\det g(x)}\Bigr)^{1/4}
u(\omega_a(t,\omega(1/t,x)))\\
&=\mathrm{e}^{-\mathrm{i}ar(x)+\mathrm{i}a^2/2t}
\Bigl(\det\frac{\partial \tau_{a(1-1/t)}(x)}{\partial x}\Bigr)^{1/2}
\Bigl(\frac{\det g(\tau_{a(1-1/t)}(x))}{\det g(x)}\Bigr)^{1/4}
u(\tau_{a(1-1/t)}(x)),
\end{align*}
where for the first factor we have used by (\ref{eq:11.5.23.7.35})
\begin{align*}
r_a(\omega(1/t,x))=r(\omega(1/t,x))-a=tr(x)-a,
\end{align*}
and for the others we have set 
\begin{align*}
\omega_a(t,\omega(1/t,x))=\tau_{a(1-1/t)}(x).
\end{align*}
In fact  $\tau_b$ is  radial translation given in spherical
coordinates by $\tau_bx(r,\sigma)= x(r+b,\sigma)$ for $r>\max(0,-b)$.
Note 
$r(\omega(1/t,x))=tr$, 
so that 
\begin{align*}
r_a(\omega_a(t,\omega(1/t,x)))
=(tr-a)/t=r_a+a(1-1/t).
\end{align*}
Hence the limit $V_a$ exists, $\mathop{\mathrm{Ran}} V_a\subseteq {\mathcal H}_{\mathrm{aux}}$ and for $x\in E$
\begin{align}\label{eq:62}
V_{a}u(x)=\mathrm{e}^{-\mathrm{i}ar(x)}
\Bigl(\det\frac{\partial \tau_{a}(x)}{\partial x}\Bigr)^{1/2}
\Bigl(\frac{\det g(\tau_{a}(x))}{\det g(x)}\Bigr)^{1/4}u(\tau_{a}(x)).
\end{align} Using \eqref{eq:62}  we see that in fact  $V_{a}$ is a unitary  map 
${\mathcal H}_{\mathrm{aux}}^{(a)}\to {\mathcal H}_{\mathrm{aux}}$.
From this unitarity property it follows that also the  limit
$V^a$ exists,  that $V^{a}$ is a unitary  map 
${\mathcal H}_{\mathrm{aux}}\to {\mathcal H}_{\mathrm{aux}}^{(a)}$ and
that $V^a=V_a^*$.

Thus we have the following relationship between the wave operators $\Omega_a$ and
$\Omega_b$:
\begin{equation}\label{eq:333}
  \Omega_a=\Omega_bV_{a}^b;\ V_a^b:=V^bV_a.
\end{equation} In fact, more generally,  the existence of $\Omega_b$ implies the
existence of $\Omega_a$ and \eqref{eq:333} is then valid (here we use
 that the limits $V^b$ and $V_a$ exist and the intertwining rule for wave operators).

\subsection{Manifold with a pole}\label{sec:manifold-with-pole}
Let us consider an ``extreme  case'' of the previous setting.
We assume, instead of Condition~\ref{cond:diffeo}:
\begin{cond}\label{cond:diffeob}
The manifold $M$ has a \textit{pole} $o$, that is, 
there exists a point $o\in M$ such that 
the exponential map:
$\exp_o\colon TM_o\to M$ 
is diffeomorphic.
\end{cond}
Note Condition~\ref{cond:diffeob} is indeed stronger than Condition~\ref{cond:diffeo},
because under Condition~\ref{cond:diffeob} we can choose any geodesic
ball for $O$. 

We consider the distance function  $r(x)=\dist(x,o)$. It is not smooth at $o$, but $r^2$ is.
Hence Condition~\ref{cond:conv}  makes sense with $r_0=0$ for the
function  $r^2$. Throughout this subsection, when we refer to
Condition~\ref{cond:conv}  we mean Condition~\ref{cond:conv} with $r_0=0$.

Define the free propagator $U(t)$, $t>0$, by
\begin{align*}
U(t)&=\mathrm{e}^{\mathrm{i}K(t,\cdot)}\mathrm{e}^{-\mathrm{i}\frac{\ln t}{2}A}
\end{align*}
with $K$ and $A$ given by (\ref{eq:59})   and (\ref{eq:6}),
respectively, in terms of the above $r^2$.
Then $\mathrm{e}^{-\mathrm{i}\frac{\ln t}{2}A}$ is the geodesic dilation with respect to $o$,
and we have the formula
\begin{align*}
U(t)u(x)=\mathrm{e}^{\mathrm{i}r(x)^2/2t}
\exp\Bigl(\int^t_1\tfrac{1}{4s}(-\triangle r^2)(\omega(s,x))\,\mathrm{d}s\Bigr)u(\omega(t,x)),
\end{align*}
where 
\begin{align*}
\omega(t,x)=\exp_o\bigl[\tfrac{1}{t}(\exp_o)^{-1}(x)\bigr]
\mfor (t,x)\in (0,\infty)\times M.
\end{align*}
\begin{thm}\label{thm:10.5.31.23.23b}
Suppose Conditions \ref{cond:diffeob} and  
\ref{cond:conv}--\ref{cond:10.6.1.16.24}.
Then there exist the strong limits
\begin{align*}
\Omega_+=\slim_{t\to+\infty} \e^{\i tH}U(t),\mand
\widetilde\Omega_+=\slim_{t\to+ \infty}U(t)^*\e^{-\mathrm{i}tH}P_{\mathrm{c}},
\end{align*}
where $P_{\mathrm{c}}$ is the projection onto 
${\mathcal H}_{\mathrm{c}}(H)=\chi_{(0,\infty)}(H)\vH$,
and the wave operator $\Omega_+$ is complete, i.e. 
\begin{equation*}
  \widetilde\Omega_+=\Omega_+^*,\;\;
\Omega_+^*\Omega_+=I \mand \Omega_+\Omega_+^*=P_{\mathrm{c}}.
\end{equation*}
\end{thm}

The result follows by combining  Theorem~\ref{thm:10.5.31.23.23} and Subsection~\ref{subsec:Uniqueness}.
In fact the arguments in Subsection~\ref{subsec:Uniqueness} extend
and are  valid including the   ``degenerate'' situation $O=\{o\}$
(even though this set  is not open).

Finally we write down the corresponding corollaries:
Noting 
\begin{equation*}
  \Omega_+^*H\Omega_+=M_f;\;f=2^{-1} r(\cdot)^2,
\end{equation*}
we have 
\begin{corollary}[Spectrum]\label{cor:asymwbb} The
  singular continuous spectrum of $H$ is absent, i.e. 
$\sigma_{\mathrm{sc}}(H)=\emptyset$, and the continuous spectrum
$\sigma_{\rm c}(H)=\sigma(H_{\rm c})=[0,\infty)$.
\end{corollary}

\begin{corollary}[Asymptotic observables]\label{cor:asymwbbb}
In the space ${\mathcal H}_{\mathrm{c}}(H)$ there exists the $*$-representation
\begin{equation*}
   \omega_{\infty}^+  =\s -C_{\mathrm{c}}(M)-\lim_{t\to
   +\infty}\e^{\i tH}\omega(t,\cdot )\e^{-\i tH}.
\end{equation*} In particular the asymptotic speed
\begin{equation*}
   r(\omega_{\infty}^+) = \s -C_{\mathrm{c}}(\R)-\lim_{t\to
   +\infty}\e^{\i tH}\;\tfrac{r(\cdot )}{t}\;\e^{-\i tH}
\end{equation*} exists as a self-adjoint operator on $\vH_{\mathrm{c}}(H)$. This
operator is positive   with zero kernel.

Moreover, for all $\phi\in C_{\mathrm{c}}(M)$
\begin{equation*}
  \phi(\omega_{\infty}^+)  =\Omega_+M_{\phi}\Omega_+^*\mand H_{\mathrm{c}}=2^{-1}r(\omega_{\infty}^+ )^2.
\end{equation*}
\end{corollary}

\begin{remark}\label{remarks:introb}
Under Condition~\ref{cond:diffeob}
a sufficient condition for Condition \ref{cond:conv} is the
  following:
Suppose there exists
    $\delta\in (0,1)$ such that the radial
    curvature $R=R(\dot x,\cdot, \dot x, \cdot)$
    satisfies the upper bound
    \begin{equation}
      \label{eq:21}
      R\leq \tfrac {1-\delta^2}{4r^2}g.
    \end{equation} This is along any unit-speed  geodesic $x(\cdot)$ emanating
    from the pole $o$ of Condition \ref{cond:diffeob} and with
    $r=r(x)$. (Alternatively $R_{ij}=(\nabla r)_k(\nabla
r)^lR^k{}_{ilj}$ in terms of  the curvature tensor as defined in
\cite{Jost,Mi}.)

Then  by a standard comparison argument, see for
    example \cite[proof of Theorem 4.1.1]{GLLT},  indeed \eqref{eq:1}
    holds true with this $\delta$. In particular if $M$ has non-positive sectional
    curvatures \eqref{eq:1} is valid for $\delta=1$. For these
     considerations Condition~\ref{cond:10.9.2.11.13} 
     is irrelevant. Note that \eqref{eq:21} involves second order
     derivatives of  the metric. In some  principal examples, see Subsubsection
     \ref{sec:10.12.25.18.37}, we shall use a different criterion
     involving  derivatives of  the metric up to first order only.  
\end{remark}

\section{Geometric setting considerations}\label{sec:geometric-setting
  considerations}
We shall  explore the generality and limitations of our conditions in terms of
various examples. The fact that these conditions are invariant under
change of variables will facilitate  the construction of examples. Secondly
we shall explore the consequences of our conditions in Classical
Mechanics. 

\subsection{Examples}\label{Examples} We give various examples. For
convenience we assume Condition~\ref{cond:diffeob} instead of Condition~\ref{cond:diffeo},
and take henceforth $r_0=0$ in Condition~\ref{cond:conv} and $V=0$ in Condition \ref{cond:10.6.1.16.24}.

\subsubsection{Warped product manifold}\label{sec:10.12.25.5.35}
Under Condition~\ref{cond:diffeob} 
we can write
\begin{align*}
g=\mathrm{d}r\otimes\mathrm{d}r+g_{\alpha\beta}(r,\sigma)\mathrm{d}\sigma^\alpha\otimes\mathrm{d}\sigma^\beta;&&
g_{rr}=1,\quad g_{r\alpha}=g_{\alpha r}=0,
\end{align*}
where $\sigma^\alpha$ are local coordinates on the geodesic unit sphere $S_1$ and the Greek indices run on $2,\dots,d$.
A \textit{warped product manifold} is a  connected complete
Riemannian manifold fulfilling  Condition~\ref{cond:diffeob} with a 
Riemannian metric of the form
\begin{align*}
g=\mathrm{d}r\otimes\mathrm{d}r+f(r)h_{\alpha\beta}(\sigma)\mathrm{d}\sigma^\alpha\otimes\mathrm{d}\sigma^\beta
\end{align*}
in the geodesic spherical coordinates. Note that this in particular means (due to a  regularity
consideration at the
pole $o$) that
$h$ is the standard 
  Euclidean metric density  of the unit sphere and that $\lim_{r\to
    0}r^{-2}f(r)=1$. In the framework of  Condition~\ref{cond:diffeo} such restriction on $h$ is not
  needed.

Let us assume $(M^d,g)$ is a warped product manifold.
Then, if we set $f=\mathrm{e}^{2\varphi}$, \eqref{eq:10.2.8.3.55} is equivalent to 
\begin{align}
2r\varphi'\ge 1+\delta,\label{eq:10.12.21.20.45}
\end{align}
and \eqref{eq:10.9.2.23.19} to
\begin{equation}
  \label{eq:40}
  |(r\varphi')'|\le C\langle r\rangle^{-1/2-\kappa}.
\end{equation}
In fact, by direct computations,
\begin{equation}\label{eq:9}
  (\nabla^2 r^2)_{rr}=2,\qquad
(\nabla^2 r^2)_{r\alpha}=(\nabla^2 r^2)_{\alpha r}=0,\qquad 
(\nabla^2 r^2)_{\alpha\beta}=r f'h_{\alpha\beta}.
  \end{equation} 
Clearly the lower bound  \eqref{eq:10.12.21.20.45} results from 
\eqref{eq:9}. Similarly the bound \eqref{eq:40} results by taking
the trace of \eqref{eq:9}, cf. \eqref{eq:10},  to obtain that  $\triangle
r^2=2+2(d-1)r\varphi'$,  and then noting that this quantity is radially symmetric.

We  see that the inequalities  \eqref{eq:10.12.21.20.45}  and \eqref{eq:40}  allow, for example, 
\begin{align*}
f_{1,\mu}(r)=r^2\langle r\rangle^{2\mu},\  \mu\ge (\delta-1)/2,&&
f_{2,\nu}(r)=r^2\e^{-2}\exp{}(2\langle r\rangle^{\nu}), \ 0\le \nu\le 1/2-\kappa.
\end{align*}
The Euclidean space corresponds to $f_{1,0}(r)=f_{2,0}(r)=r^2$.

\subsubsection{Ultra-long-range perturbation of Euclidean space}
Though we have formulated our conditions  in a coordinate invariant way, 
our first motivation was the
  example $M=\mathbb{R}^d$ with a Riemannian metric $g$ satisfying Conditions~\ref{cond:diffeob},
\ref{cond:conv} and 
\begin{cond}\label{cond:10.12.25.3.36}
There exists $c>0$ such that for the standard coordinates $x$
\begin{align*}
g\ge c\delta_{ij}\mathrm{d}x^i\otimes\mathrm{d}x^j
\end{align*}
and that for $r=\mathop{\mathrm{dist}}_g(x,0)$
\begin{align*}
|\partial^\alpha_x g_{ij}|&{}\le C_\alpha\langle x\rangle^{-|\alpha|}\mfor |\alpha|\leq 2,\\
|\partial_x^\alpha (\nabla^2 r^2)_{ij}|&{}\le C_\alpha \langle x\rangle^{-|\alpha|}
\mfor |\alpha|\le 1.
\end{align*}
\end{cond}
The Condition~\ref{cond:10.12.25.3.36} is
stronger than Condition~\ref{cond:10.9.2.11.13}. Note  also that 
 Condition~\ref{cond:10.12.25.3.36} is  manifestly  not coordinate
 invariant 
requiring $g$ to be comparable with the Euclidean
metric.

An  example of a model
  satisfying Conditions~\ref{cond:diffeob},
\ref{cond:conv} and \ref{cond:10.12.25.3.36} is
  given as follows: Let $m$ be a  real symmetric $d\times
  d$--matrix-valued function on $\R^d$. Suppose in addition that all
  entries $m_{ij}\in C^\infty(\R^d)$ and obey
  \begin{equation}\label{eq:m17}
   |\partial^\alpha_x m_{ij}|\le C_\alpha\langle
   x\rangle^{-|\alpha|}\mfor |\alpha|\leq 3. 
  \end{equation} Then for any $\epsilon\in \R$ with  $|\epsilon|$
  being sufficiently
  small  the metric $g$ given as a matrix by 
  $g_{ij}=\delta_{ij} +\epsilon m_{ij}$ fulfills
  Conditions~\ref{cond:diffeob}, 
\ref{cond:conv} and \ref{cond:10.12.25.3.36}. We refer to \cite{CS} for details. 

In fact
  there is a more general example from  \cite{CS}: 
Take any  ``unperturbed''
  metric $g$ on $\R^d$ obeying 
  \begin{align}
g\ge c\delta_{ij}\mathrm{d}x^i\otimes\mathrm{d}x^j,&&
|\partial^\alpha_x g_{ij}|&{}\le C_\alpha\langle x\rangle^{-|\alpha|}\mfor |\alpha|\leq 3
\end{align}
and identified in terms of the Euclidean metric on $\R^d$ as a matrix of
  the form (for $x\neq0$)
\begin{equation} 
\label{ortdec22}
G(x)=P+P_\perp G(x) P_\perp,
\end{equation}
where $P$ denotes, in the Dirac notation,
the orthogonal projection $P=P(\hat x)=|\hat x\rangle\langle\hat x|$
parallel to $\hat x=x/|x|$ and $P_\perp=P_\perp(\hat x)=I-P$
the orthogonal projection onto $\{\hat x\}^\perp$. Suppose in addition
that
\begin{equation}
\label{ortdec2}
P_\perp \parb {(1-\delta)G(x) +x \cdot \nabla G(x)} P_\perp\geq 0.
\end{equation} Then a computation shows that the conditions of
\cite[Theorem 1.4 ii)]{CS} as well as Condition
\ref{cond:conv} (with this  $\delta$ in Condition
\ref{cond:conv} and with $\bar c:=(1+\delta)/2$ in \cite[(1.13)]{CS})
are fulfilled. In fact using \eqref{ortdec22} we compute 
\begin{equation*}
  \nabla^2 r^2(x)(y,y)=2yG(x)y +yP_\perp\nabla G(x)\cdot xP_\perp y,
\end{equation*} showing  the equivalence of   Condition
\ref{cond:conv} and \eqref{ortdec2} for a metric of the form
\eqref{ortdec22}. (Note at this point the consistency with
\eqref{eq:10.12.21.20.45}.) 

If we again
let $m$ be given by \eqref{eq:m17} and similarly define 
$(g_\epsilon)_{ij}=g_{ij}+\epsilon m_{ij}$ then a computation using \cite[Theorems 1.4
ii) and 1.6]{CS} shows that indeed  $g_\epsilon$ for any sufficiently
  small $|\epsilon|$  fulfills Conditions~\ref{cond:diffeob},
\ref{cond:conv} and \ref{cond:10.12.25.3.36}. For some examples constructed in this way
  we refer to Subsubsection~\ref{sec:10.12.25.18.37}. As the reader
  will see  the geometric invariance  is
  exploited explicitly.

\subsubsection{Conformally flat manifold}\label{sec:10.12.25.18.37}
The Laplace-Beltrami operator, the comparison dynamics
\eqref{eq:10.5.4.22.21} and Conditions
\ref{cond:diffeo}--\ref{cond:10.9.2.11.13} (as well as Condition
\ref{cond:10.6.1.16.24}) are    cleanly
geometrically  invariant,
while Condition~\ref{cond:10.12.25.3.36} is not that appealing. One way
to 
circumvent this for  given  $(M,g_M)$, $M$ being connected, complete  and
$d$-dimensional, is by postulating the  existence of a  diffeomorphism $\Psi:M\to
\R^d$ with the property that 
\begin{equation}
  \label{eq:20}
  g:=(\Psi^{*})^{-1}g_M\text{ obeys Condition \ref{cond:10.12.25.3.36}.
  }
\end{equation} Clearly Conditions
\ref{cond:diffeob}, \ref{cond:conv} and \eqref{eq:20} constitutes an
 invariant theory.

We will in this subsubsection give
an 
 example of the how to use \eqref{eq:20} concretely. Our discussion is
 based on \cite[Section 7]{CS}. Consider  a  radial function
$V=V(z)=V(|z|)$ of class $C^\infty$  on $\R^d$, $d\geq 2$,   for which
 there are constants  $ \mu\in (-\infty,2)$, $a,A, \sigma>0$  and $\delta\in
 (0,1]$ such that
 \begin{subequations}
   \begin{align}
\label{bound1v}
-A \langle z\rangle^{-\mu}\leq V(z)&\leq -a \langle z\rangle^{-\mu},
\\
z\cdot \nabla V(z)+2V(z)&\leq \sigma V(z)\label{bound1vbb},\\
\partial^\alpha V(z)&=O\parb{\inp{z}^{-(\mu+|\alpha|)}}\mfor |\alpha|\leq 3,\label{bound1vbbcc}
\\(1+\delta)& \leq \inf_{r>0} h(r);\,
h(r):=\frac{\big (2+\tfrac{-rV'(r)}{-V(r)}\big )\int_0^r\sqrt{-V(s)}\d s}{\sqrt{-V(r)}r}.\label{bound1vbbqq}
\end{align} 
 \end{subequations}

Note that for any $ \mu\in [0,1)$ obviously the function $V(z)=-\inp{z}^{-\mu}$ is an example in
this class. A more careful (but elementary) consideration shows that  this is  the case for any $ \mu\in
(-\infty,4/3)$. It could be true for any $ \mu\in
(-\infty,2)$ (note that $h(0)=h(\infty)=2$). If $W=W(z)$ is of
class $C^\infty$  on $\R^d$, possibly non-radial, we are interested in
studying the metric (conformally) generated by $V_\epsilon=V+\epsilon W$ for
$|\epsilon|\geq 0$
sufficiently small, that is the metric
\begin{equation}
  \label{eq:22}
  -2V_\epsilon\,\d z^2.
\end{equation} We shall impose a  condition on $W$  similar to \eqref{bound1vbbcc},
\begin{equation}
  \partial^\alpha W(z)=O\parb{\inp{z}^{-(\mu+|\alpha|)}}\mfor |\alpha|\leq 3.\label{bound1vbbccb}
\end{equation} Under the conditions
\eqref{bound1v}--\eqref{bound1vbbcc} and \eqref{bound1vbbccb}  a
diffeomorphism $\Psi$ as in \eqref{eq:20} is constructed in
\cite[Subsections 7.1--2]{CS}
so that in the new coordinates in fact  Conditions
\ref{cond:diffeob} and \ref{cond:10.12.25.3.36} (and therefore Conditions
\ref{cond:diffeob}  and \ref{cond:10.9.2.11.13}) are fulfilled. The new condition \eqref{bound1vbbqq} is introduced, as we will
show below, to verify the remaining Condition
\ref{cond:conv} (in our discussion the potential  in \eqref{eq:10.5.4.3.24} is for simplicity taken absent). Following  \cite{CS} we  define $\Psi$ by specifying
its inverse,
\begin{equation}
  \label{eq:23}
  \Psi^{-1}(x)=\exp_0\big (x/(-2V(0)\big),
\end{equation} where the exponential mapping is defined in terms of
the unperturbed metric, i.e. by  \eqref{eq:22} with
$\epsilon=0$. Concretely 
\begin{equation}
  \label{eq:23b}
  x=\Psi(z)=\rho(|z|)\tfrac{z}{|z|}\text{ where }\rho(r)=\int_0^r\sqrt{-2V(s)}\d s.
\end{equation} Letting $r=r(\rho)$ denote the inverse of this
function $\rho$ we define 
\begin{equation}
  \label{eq:25}
  f(\rho)=\frac{\sqrt{-2V(r(\rho))} r(\rho)}{\rho},
\end{equation} and  
we have 
\begin{equation}
  \label{eq:24}
  g_\epsilon:=(\Psi^{*})^{-1}\parb{-2V_\epsilon\,\d^2z}= P+f^2P_\perp +O(\epsilon),
\end{equation} where $P$ and $P_\perp$ are given as in
\eqref{ortdec22}. It remains (for Condition
\ref{cond:conv}) to show that indeed \eqref{ortdec2} holds for the
unperturbed part, $G:=P+f^2P_\perp$, given the condition
\eqref{bound1vbbqq}: We compute
\begin{equation*}
  \rho\tfrac{\d}{\d
    \rho}f^2+(1-\delta)f^2=f^2\parbb{\parb{2+\frac{-rV'(r)}{-V(r)}}/f-1-\delta}=f^2\parb{h(r)-1-\delta};
\end{equation*} whence indeed by \eqref{bound1vbbqq}
\begin{equation*}
  P_\perp \parb {(1-\delta)G(x) +x \cdot \nabla G(x)} P_\perp\geq 0,
\end{equation*} and therefore Condition
\ref{cond:conv} holds for the metrics $g_\epsilon$ and
$-2V_\epsilon\,\d z^2$ for any slightly smaller $\delta >0$ 
provided $|\epsilon|\geq0$ is taken small enough. In particular our results
apply to the metric $-2V_\epsilon\,\d z^2$ although this example from
the outset does
not conform with Condition~\ref{cond:10.12.25.3.36} 
 (unless $\mu=0$).

We refer the reader to \cite[Subsection 7.3]{CS} for an example with
$\mu=0$ in  the previous scheme for which  the geodesics
of the perturbed metric emanating from $0\in \R^d$  are  
 attracted to logarithmic spirals.

\subsection{Counter examples, borderlines  of conditions}\label{Generalizations}
We construct warped product manifolds to illustrate the optimality of the
conditions of
Theorem \ref{thm:10.5.31.23.23}.
\begin{proposition}
  \label{prop:count-exampl-bord} Suppose Condition
  \ref{cond:diffeob}. Suppose that exactly one of the conditions
  $\delta>0$, $\kappa>0$ and $\eta>0$ in Conditions
  \ref{cond:conv}--\ref{cond:10.6.1.16.24}  is replaced 
  by either $\delta=0$, $\kappa=0$ or $\eta=0$, respectively.  Then
  the exists a warped product manifold fulfilling this slightly more
  general set of conditions for which not all of the analogous conclusions of
  Theorem \ref{thm:10.5.31.23.23} are  true.
\end{proposition}
 
In the case of $\delta=0$ we can choose 
the density factor  $f(r)=r^2\inp{r}^{-1}$ in Subsubsection \ref{sec:10.12.25.5.35}.
In the case of $\kappa=0$ we can choose 
the density factor $f(r)=r^2\e^{-2}\e^{2\sqrt{\inp{r}}}$. While for
$\eta=0$ we can choose 
the density factor  $f(r)=r^2$ (Euclidean model) and
$V=c\inp{r}^{-1}$, $c\neq 0$.

To see this we need some preparation. We introduce the Hilbert space
$\tilde \vH=L^2(\R_+,\vG,\d r)$ where $\vG =L^2(S_1,\d \sigma)$ where $S_1$
is the unit sphere in $\R^d$ and $\d \sigma$ the induced Euclidean
measure. For any warped product model $f=\e^{2\phi}$  we introduce a
unitary operator
$\vM:\vH\to \tilde \vH$ by
\begin{equation*}
  \tilde u=\vM u=f(r)^{\frac{d-1}{4}}u=f(1)^{\frac{d-1}{4}}\e^{\int_1^r\triangle r \,\d r/2}u.
\end{equation*}  
We have the formula in spherical coordinates $(r,\sigma)$ 
\begin{equation*}
  \tilde U(t)\tilde u:=\vM U(t)\vM ^{-1}\tilde u=\e^{\i \tfrac{r^2}{2t}}t^{-1/2}\tilde u(\tfrac{r}{t},\sigma).
\end{equation*} Note also the formula 
\begin{equation*}
  \tilde H \tilde u:=\vM H\vM ^{-1}\tilde u=\parb{\tfrac 12
    p_r^2+\tilde V-\tfrac 12 f^{-1}\triangle _1}\tilde u,
\end{equation*} where
\begin{equation*}
   p_r=-\i\tfrac{\d}{\d r},\quad\tilde V=V+\tfrac 18 (\triangle r)^2 +\tfrac 14 \partial
    \triangle r,
\end{equation*}
 and $\triangle_1$ denotes the Laplace-Beltrami
operator on $\vG$. 

The generator of $\tilde U(t)$ is given  by
\begin{equation*}
  \tilde G(t)=\tfrac{1}{2}p_r^2-\tfrac{1}{2}\bigl(p_r-\frac{r}{t}\bigr)^2,
\end{equation*}  cf. 
Remark \ref{remarks:intro} \ref{item:3}.  Whence we have for all
$\tilde u\in C^\infty_c(\R_+\times S_1)$
\begin{equation*}
  \parb {\tilde H-\tilde G(t)}\tilde U(t)\tilde u=t^{-1/2}\e^{\i \tfrac{r^2}{2t}}\parb {-\tfrac{1}{2}p_r^2+\tilde V-\tfrac 12 f^{-1}\triangle _1}\tilde u(\tfrac{r}{t},\sigma).
\end{equation*} Given the  Conditions
\ref{cond:conv}--\ref{cond:10.6.1.16.24} for the model, and therefore
  \eqref{eq:49}, \eqref{eq:10.12.21.20.45} and \eqref{eq:40}, the Cook method
 and this computation yields the existence of the limit
\begin{equation}
  \label{eq:58}
  \tilde \Omega\tilde u=\lim_{t\to \infty}\e^{\i \tilde H t}\tilde
  U(t)\tilde u;\; \tilde u\in \tilde \vH.
\end{equation} Moreover this argument does not work if $\delta=0$ in Condition
\ref{cond:conv} (since then $f^{-1}$ might not decay fast enough), or
if $\kappa=0$ in Condition \ref{cond:10.9.2.11.13} (since  then
$(\triangle r)^2$ might not decay fast enough) nor if $\eta=0$ in
Condition \ref{cond:10.6.1.16.24} (since  then
$V$ might not decay fast enough). This provides some   intuition about
Proposition \ref{prop:count-exampl-bord}. 

To come closer to a 
proof of Proposition \ref{prop:count-exampl-bord} let us note that these borderline cases can be ``repaired'' by
modified evolutions in the spirit of the Dollard evolution for Schr\"odinger
operators. Thus, for the example $f(r)=r^2\inp{r}^{-1}$ the factor
$f^{-1}\approx r^{-1}$, and if we take $\tilde u=v(r)Y_l(\sigma)$ for
a spherical harmonic $Y_l$ we have
\begin{equation*}
  -\tfrac 12 f^{-1}\triangle _1\tilde u(\tfrac{r}{t},\cdot)\approx   cr^{-1}\tilde u(\tfrac{r}{t},\cdot);\;c=c(l)=-\tfrac 12 (l+d/2-1).
\end{equation*} This  motivates us to introduce 
\begin{equation*}
  \tilde U_l(t)\tilde u:=\e^{\i \parb{\tfrac{r^2}{2t}-c\tfrac t{\inp{r}}
    \ln \inp{r}}}t^{-1/2}\tilde u(\tfrac{r}{t},\sigma),
\end{equation*} whose generator is 
\begin{equation*}
  \tilde
  G_l(t)=\tfrac{1}{2}p_r^2-\tfrac{1}{2}\bigl(p_r-\frac{r}{t}\bigr)^2+c\tfrac
  {\ln \inp{r}}{\inp{r}^3}+c\tfrac {r^2}{\inp{r}^3}.
\end{equation*} Note that $c\tfrac {r^2}{\inp{r}^3}\approx cr^{-1}$. Whence, by the arguments above for this example,  we obtain the existence of the limit
\begin{equation}
  \label{eq:58b}
  \tilde \Omega_l\tilde u=\lim_{t\to \infty}\e^{\i \tilde H t}\tilde
  U_l(t)\tilde u;\; \tilde u\in \tilde \vH_l:=L^2(\R_+)\otimes Y_l.
\end{equation}
We note the property 
\begin{equation*}
  \tilde \Omega_l^*\tilde H\tilde \Omega_l=M_\lambda,
\end{equation*} where $M_\lambda$ denotes multiplication by the
function  $r\to r^2/2$, $v(r)\otimes Y_l\to \tfrac{r^2}{2}v(r)\otimes
Y_l$. (The reader may at this point consult the end of
Subsection \ref{lemma:10.5.31.23.24}.)
 In particular $\Ran \vM^{-1}\tilde \Omega_l\subseteq \vH_c(H)$.

\proof[Proof of Proposition \ref{prop:count-exampl-bord}] First we
continue our discussion of the  example $f(r)=r^2\inp{r}^{-1}$  for
which $\delta=0$. Suppose on the contrary that the conclusions of
  Theorem \ref{thm:10.5.31.23.23} are all true for this example. Let
   $ \Omega_+$ be given accordingly and $\tilde \Omega_l$ be given by \eqref{eq:58b}. We
  derive a contradiction by taking an arbitrary nonzero $\tilde v\in
  L^2(\R_+)$, define $\tilde u=\tilde v\otimes Y_l$ with any $l$ for
  which $c=c(l)\neq 0$ and compute 
  \begin{align*}
    \Omega^*_+\vM^{-1}\tilde \Omega_{l}\tilde
    u&=\Omega^*_+P_c\vM^{-1}\tilde \Omega_{l}\tilde u=\lim_{t\to \infty}U(t)^*
  \vM^{-1}\tilde U_l(t)\tilde u\\ &=\vM^{-1}\lim_{t\to \infty}\parb{\e^{-\i \tfrac c{r}
    \ln t}
  \tilde w}\otimes Y_l;\;\tilde w(r)=\e^{-\i \tfrac c{r}
    \ln r}\tilde v(r).
  \end{align*} Since $\tilde w\neq 0$ and $c\neq 0$  the factor $\e ^{-\i \tfrac c{r}
    \ln t}$ does not have a  limit when applied to $\tilde w$, cf.   the Riemann-Lebesgue lemma \cite{RS}). This is a contradiction.

For  the  example $f(r)=r^2\e^{-2}\e^{2\sqrt{\inp{r}}}$, for
which $\kappa=0$, we proceed
similarly. The term $\tfrac 18 (\triangle r)^2 \approx c r^{-1}
$,  so we can repeat the above arguments.

Finally  the potential $V=c\inp{r}^{-1}$, for
which $\eta=0$, provides (with $f=r^2$) a counter example. The arguments are the same.
\endproof

\subsection{Classical Mechanics under Conditions \ref{cond:diffeob}\ and \ref{cond:conv}}\label{Classical mechanics considerations}
We outline proofs of analogues of
Theorem \ref{thm:10.5.31.23.23b}  and Corollary \ref{cor:asymwbbb} in
Classical Mechanics.
As we pointed out before  the  Classical Mechanics considerations only
require Conditions~\ref{cond:diffeo} and \ref{cond:conv}.
But, for convenience, we consider Conditions~\ref{cond:diffeob} and \ref{cond:conv} with $r_0=0$, instead.
 If we adopt Conditions~\ref{cond:diffeo} and \ref{cond:conv} not necessarily with $r_0=0$,
 then all the geodesics appearing below need to be non-trapped.  Our  proofs of Theorem \ref{thm:10.5.31.23.23} and
Corollary \ref{cor:asymw} are strongly motivated by   these considerations.

\subsubsection{Regularity of classical
  dilation}\label{sec:regul-inverse-dilat}
First we prove an estimate for the  geodesic dilation $\omega(t,x)$.
Recall 
\begin{align*}
\omega(t,x)=\exp_o \Bigl[\tfrac{1}{t}\exp_o^{-1}(x)\Bigr],&& (t,x)\in (0,\infty)\times M.
\end{align*}
 In any local coordinates $\omega$ satisfies, cf. \eqref{09.12.17.22.23},
\begin{align}
\partial_t \omega^i=-\tfrac{r}{t^2}(\grad r)(\omega)=-\tfrac{1}{2t}g^{ij}(\omega)(\partial_jr^2)(\omega),&&
\omega(1,x)=x.
\label{09.12.17.22.23b}
\end{align}
\begin{lemma} For all $(t,x)\in (0,\infty)\times M$ and independently of choice of coordinates
\begin{align}
g^{ij}(x)g_{kl}(\omega(t,x))[\partial_i\omega^k(t,x)][\partial_j\omega^l(t,x)]\le d t^{-(1+\delta)}.\label{9.12.19.1.58}
\end{align}
\end{lemma}
\proof
The left hand side of \eqref{9.12.19.1.58} is indeed independent of
coordinates. Fix $x\in M$ and choose coordinates such that $g_{ij}(x)=\delta_{ij}$. 
Consider the 
vector fields along $\{\omega(t,x)\}_{t\in\mathbb{R}}$ given by $\partial_i\omega^\bullet(t,x)$ and $\partial_j\omega^\bullet(t,x)$.
Since the Levi-Civita connection $\nabla$ is compatible with the metric, 
\begin{align*}
\tfrac{\partial }{\partial t}g_{kl}(\omega)(\partial_i\omega^k)(\partial_j\omega^l)
=\tfrac{\partial }{\partial t}\langle \partial_i\omega^\bullet,\partial_j\omega^\bullet\rangle
=\langle \nabla_{\partial_t\omega} \partial_i\omega^\bullet,\partial_j\omega^\bullet\rangle
+\langle \partial_i\omega^\bullet,\nabla_{\partial_t\omega} \partial_j\omega^\bullet\rangle.
\end{align*}
(The definition of $\nabla_{\partial_t\omega}$ is given below.)
From (\ref{09.12.17.22.23b}) it follows that
\begin{align*}
\nabla_{\partial_t \omega}\partial_i \omega^\bullet
&{}=\partial_t\partial_i \omega^\bullet+(\partial_t \omega^k)\Gamma^\bullet_{kl}\partial_i \omega^l\\
&{}=-\tfrac{1}{2t}(\partial_i \omega^k)\partial_k(g^{\bullet l}\partial_lr^2) 
-\tfrac{1}{2t}( g^{km}\partial_mr^2)\Gamma^\bullet_{kl}\partial_i \omega^l\\
&{}=-\tfrac{1}{2t}\nabla_{\partial_i \omega}(g^{\bullet l}\partial_l r^2)\\
&{}=-\tfrac{1}{2t}g^{\bullet l}(\partial_i \omega^k) (\nabla^2 r^2)_{kl}.
\end{align*}
Thus, taking summation in $i,j$, we obtain
\begin{align*}
\tfrac{\partial }{\partial t}g^{ij}(x)g_{kl}(\omega)(\partial_i\omega^k)(\partial_j\omega^l)
\le -\tfrac{1+\delta}{t}g^{ij}(x)g_{kl}(\omega)(\partial_i\omega^k)(\partial_j\omega^l).
\end{align*}
Noting $g^{ij}(x)g_{kl}(\omega)(\partial_i\omega^k)(\partial_j\omega^l)\bigr|_{t=1}=d$, 
we have (\ref{9.12.19.1.58}).
\endproof

\subsubsection{Propagation estimates}\label{sec:prop-estim} 
Set
\begin{align*}
K(t,x)=\tfrac{r^2}{2t}, &&h_0(x,\xi)=\tfrac{1}{2}g^{ij}\xi_i\xi_j,&&
w(t,x,\xi)=\tfrac{1}{2}g^{ij}(\xi_i-\partial_i K)(\xi_j-\partial_j K)
\end{align*}
for $t>0$, $x\in M$ and $(x,\xi)\in T^*M$.
\begin{lemma}
For any  Hamiltonian trajectory $(x(t),\xi(t))$ there exists $C>0$ such that 
\begin{align}
w(t,x(t),\xi(t))\le Ct^{-(1+\delta)}.
\label{9.12.19.1.57}
\end{align}
\end{lemma}
\proof
 We compute
\begin{align*}
\tfrac{\mathrm{d}}{\mathrm{d}t} w
=\tfrac{\partial}{\partial t}w +\{h_0,w-h_0\}
=\tfrac{\partial}{\partial t}w +\tfrac{\partial h}{\partial \xi} \tfrac{\partial}{\partial x}(w-h_0)
-\tfrac{\partial h}{\partial x}\tfrac{\partial}{\partial \xi} (w-h_0).
\end{align*}
By (\ref{eq:9.12.19.054}) 
\begin{align*}
\tfrac{\partial}{\partial t}w
=\frac{1}{2}g^{ij}(\partial_i g^{kl}(\partial_kK)(\partial_lK))(\xi_j-\partial_jK).
\end{align*}
Noting that by the compatibility condition $(\nabla g)^{ij}_k=0$, we have
\begin{align}
0=\partial_kg^{ij}+\Gamma^i_{kl}g^{lj}+\Gamma^j_{kl}g^{il},
\label{09.11.28.20.51}
\end{align}
so that 
\begin{align}
\partial_i g^{kl}(\partial_k K)(\partial_l K)=2(\nabla ^2 K)_{ik}g^{kl}(\partial_l K).
\label{09.11.29.3.18}
\end{align}
Thus
\begin{align*}
\tfrac{\partial}{\partial t}w
=(\partial_l K)g^{lk}(\nabla^2K)_{ki}g^{ij}(\xi_j-\partial_j K).
\end{align*}
On the other hand, by (\ref{09.11.28.20.51}) and (\ref{09.11.29.3.18}) again we have 
\begin{align*}
&{}\{h_0,w-h_0\}\\
&{}=
g^{ij}\xi_j
\Bigl[-(\partial_ig^{kl})\xi_k(\partial_l K)
-g^{kl}\xi_k(\partial_i\partial_l K)
+(\nabla ^2 K)_{ik}g^{kl}(\partial_l K)\Bigr]
+\tfrac{1}{2}(\partial_k g^{ij})\xi_i\xi_jg^{kl}(\partial_l K)\\
&{}=g^{ij}\xi_j(\Gamma^k_{im}g^{ml}+\Gamma^l_{im}g^{km})\xi_k(\partial_l K)
-g^{ij}\xi_jg^{kl}\xi_k(\partial_i\partial_l K)
+g^{ij}\xi_j(\nabla ^2 K)_{ik}g^{kl}(\partial_l K)\\
&\phantom{{}={}}{}
-
\tfrac{1}{2}(\Gamma^i_{km}g^{mj}+\Gamma^j_{km}g^{im})\xi_i\xi_j
g^{kl}(\partial_l K)\\
&{}
=g^{ij}\xi_j\Gamma^l_{im}g^{km}\xi_k(\partial_l K)
-g^{ij}\xi_jg^{kl}\xi_k(\partial_i\partial_l K)
+g^{ij}\xi_j(\nabla ^2 K)_{ik}g^{kl}(\partial_l K)\\
&{}
=-
\xi_jg^{ji}
(\nabla^2 K)_{ik}g^{kl}(\xi_l-\partial_l K).
\end{align*}
Hence, summing up and using Condition~\ref{cond:conv}, we obtain
\begin{align*}
\tfrac{\mathrm{d}}{\mathrm{d}t} w
=-(\xi_l-\partial_l K)g^{lk}(\nabla^2K)_{ki}g^{ij}(\xi_j-\partial_j K)
\le -\tfrac{1+\delta}{t}w.
\end{align*}
\endproof

\begin{proposition} For any  geodesic $x(t)$ there exists the limit
\begin{equation}\label{aswww}
\omega_\infty=\lim_{t\to\infty}\omega(t,x(t)).
\end{equation}
\end{proposition}
\proof
Due to  the flow equation  \eqref{09.12.17.22.23b} we have  the group property
$\omega(t,\omega(s,x))=\omega(ts,x)$. 
Differentiate 
\begin{equation*}
  \omega(t,\omega(s,x))=\omega(s,\omega(t,x))
\end{equation*} in $t$, and use then \eqref{09.12.17.22.23b} to obtain
\begin{equation*}
  \partial_t \omega^i(t,\omega(s,x))=-(\partial_k
  \omega^i)(t,\omega(s,x))g^{kl}(\omega(s,x))(\partial_l K)(t,\omega(s,x)).
\end{equation*}
Putting $s=1$, we obtain
\begin{equation}
  \partial_t \omega^i(t,x)=-g^{kl}(x)(\partial_k
  \omega^i)(t,x)(\partial_l K)(t,x).
\label{09.12.19.2.3}
\end{equation}

By applying first (\ref{09.12.19.2.3}), and then  (\ref{9.12.19.1.58}) and (\ref{9.12.19.1.57}), we obtain
\begin{align*}
g_{ij}\dot{\omega}^i\dot{\omega}^j
&{}=g_{ij}[\partial_t \omega^i+(\partial_\xi h_0)(\partial_x \omega^i)][\partial_t \omega^j+(\partial_\xi h_0)(\partial_x \omega^j)]\\
&{}=g_{ij}g^{kl}(\partial_k\omega^i)(\xi_l-\partial_lK)g^{mn}(\partial_m\omega^j)(\xi_n-\partial_nK)\\
&{}\le Ct^{-2(1+\delta)},
\end{align*}
and the assertion follows.
\endproof

\subsubsection{Mourre estimate}\label{sec:10.12.22.12.20}
We also note that the classical Mourre estimate holds.
Since the geodesics equation is given by 
$\ddot{x}^i+\Gamma^{i}_{jk}\dot{x}^j\dot{x}^k=0$,
the following result is true.
\begin{lemma}
For any geodesic $x(t)$ the following inequality holds:
\begin{align*}
\tfrac{\mathrm{d}^2}{\mathrm{d}t^2}r^2
\ge 2(1+\delta)h_0.
\end{align*}
\end{lemma}

\section{Reduction of the proof of Theorem~\ref{thm:10.5.31.23.23}}
\subsection{Reduction to existence of localization operators}
Since we do not have enough regularities for the derivatives of $g$ and $\nabla^2r^2$,
the Cook-Kuroda method does not apply even for the existence part of Theorem~\ref{thm:10.5.31.23.23}.
We shall prove Theorem \ref{thm:10.5.31.23.23} in a symmetric manner for
the existence and the completeness parts.
In this subsection we reduce the proof to the construction of
$Q_{\mathrm{f}}(t)$ and $Q_{\mathrm{p}}(t)$ which are 
time-dependent localization operators for the {\it free} and the  {\it
  perturbed} dynamics, respectively.

We denote the time-dependent generator of $U(t)$ by $G(t)$, i.e.,
\begin{align*}
\tfrac{\mathrm{d}}{\mathrm{d}t}U(t)=-\mathrm{i}G(t)U(t).
\end{align*} It will not be important to known the domain of the
generator but rather a convenient subspace. For that  we observe that
$U(t)$ and $U(t)^{-1}$ preserve the subspace $C_{\rm c}^\infty(M)\subseteq
\vH$,  and hence $C_{\rm c}^\infty(M)\subseteq \vD(G(t))$, and that the
propagator acting on this subspace is explicitly given as follows
(recall $H_0=-\frac12 \triangle$):
\begin{subequations}
\begin{align}
G(t)&=-\partial_tK+\mathrm{e}^{iK}\tfrac{1}{2t}A\mathrm{e}^{-iK}=H_0-W(t)-\alpha(t)
;\label{10.5.14.15.19}\\
W(t)&=\tfrac{1}{2}(p_i-\partial_i K)^*g^{ij}(p_j-\partial_j K)=\e ^{\i K}H_0\e ^{-\i K},\\
\alpha(t)&=\partial_tK+\tfrac{1}{2}g^{ij}(\partial_iK)(\partial_jK)
\end{align}\end{subequations}
Note by (\ref{eq:9.12.19.054}) that $\alpha(t)\equiv 0$ on $E$, and thus 
\begin{align}
\forall n\in\N :|\alpha|=O(t^{-2}\langle r\rangle^{-n})\mon M.\label{eq:11.5.5.13.48}
\end{align}

For all practical purposes (in particular for stating
Lemma \ref{lemma:10.5.31.23.24} below) we can consider
\eqref{10.5.14.15.19} as a definition of a symmetric  operator $ G(t)$ on the
domain $\vD(H_0)\cap\vD\parb{H_0\e ^{-\i
    K}}=\vD(H_0)\cap\vD\parb{W(t)}$.

As shown at the end of the subsection Theorem~\ref{thm:10.5.31.23.23}
is a consequence of the following two lemmas:
\begin{lemma}\label{lemma:10.5.31.23.24}
Let $0<\mu<M<\infty$.
Then there exists
a weakly differentiable $Q_{\mathrm{f}}\colon [1,\infty)\to{\mathcal B}_{\rm sa}({\mathcal H})$
such that $\|Q_{\mathrm{f}}(t)\|_{{\mathcal B}({\mathcal H})}\le 1$ and for some $\delta'>0$
\begin{enumerate}[i)]
\item\label{item:10.5.31.23.25}
\begin{align*}
\slim_{t\to\infty}(I-Q_{\mathrm{f}}(t))U(t)\chi_{[\mu,M]}(r^2)P_{\mathrm{aux}}=0,
\end{align*}
where $\chi_{[\mu,M]}$ is the characteristic function for $[\mu,M]$ and 
$\chi_{[\mu,M]}(r^2)$ denotes the multiplier,
\item\label{item:10.5.31.23.26} The operators 
$G(t)Q_{\mathrm{f}}(t)$ and $Q_{\mathrm{f}}(t)G(t)$ are bounded, and
the Heisenberg derivative of $Q_{\mathrm{f}}(t)$  with respect to $G(t)$
is non-negative 
modulo $O_{{\mathcal B}({\mathcal H})}(t^{-1-\delta'})$:
\begin{align*}
\exists R(t)=O_{{\mathcal B}({\mathcal H})}(t^{-1-\delta'})\quad \mbox{s.t.}\quad  
\mathrm{D}_{G(t)}Q_{\mathrm{f}}(t)
=\tfrac{\mathrm{d}}{\mathrm{d}t} Q_{\mathrm{f}}(t)+\mathrm{i}[G(t),Q_{\mathrm{f}}(t)]\ge R(t),
\end{align*}
\item\label{item:10.5.31.23.27} The operators
$(W(t)+\alpha(t)+V)Q_{\mathrm{f}}(t)$ and $Q_{\mathrm{f}}(t)(W(t)+\alpha(t)+V)$ are 
$O_{{\mathcal B}({\mathcal H})}(t^{-1-\delta'})$.
\end{enumerate}
\end{lemma}
\begin{lemma}\label{lemma:10.3.23.17.20}
Let $E\in (0,\infty)$.
If $e>0$ is sufficiently small, then there exists 
a weakly differentiable $Q_{\mathrm{p}}\colon [1,\infty)\to{\mathcal B}_{\rm sa}({\mathcal H})$
such that $\|Q_{\mathrm{p}}(t)\|_{{\mathcal B}({\mathcal H})}\le 1$ and for some $\delta'>0$
\begin{enumerate}[i)]
\item\label{item:10.3.23.17.21}
\begin{align*}
\slim_{t\to\infty}(I-Q_{\mathrm{p}}(t))\e^{-\mathrm{i}tH}\chi_{[E-e,E+e]}(H)=0,
\end{align*}
\item\label{item:10.3.23.17.22} The operators
$HQ_{\mathrm{p}}(t)$ and $Q_{\mathrm{p}}(t)H$ are bounded, and 
\begin{align*}
\exists R(t)=O_{{\mathcal B}({\mathcal H})}(t^{-1-\delta'})\quad \mbox{s.t.}\quad  
\mathrm{D}_HQ_{\mathrm{p}}(t)=\tfrac{\mathrm{d}}{\mathrm{d}t} Q_{\mathrm{p}}(t)+\mathrm{i}[H,Q_{\mathrm{p}}(t)]\ge R(t),
\end{align*}
\item\label{item:10.3.23.17.23} The operators
$(W(t)+\alpha(t)+V)Q_{\mathrm{p}}(t)$ and
$Q_{\mathrm{p}}(t)(W(t)+\alpha(t)+V)$  are 
$O_{{\mathcal B}({\mathcal H})}(t^{-1-\delta'})$.
\end{enumerate}
\end{lemma}
Now we deduce Theorem~\ref{thm:10.5.31.23.23} from Lemmas~\ref{lemma:10.5.31.23.24} and \ref{lemma:10.3.23.17.20}.
The existence and the completeness parts are completely the same and we discuss only the 
existence part.
From Lemma~\ref{lemma:10.5.31.23.24} \ref{item:10.5.31.23.26} and \ref{item:10.5.31.23.27}  
the following statement follows, which combined with 
 Lemma~\ref{lemma:10.5.31.23.24} \ref{item:10.5.31.23.25}  and  a
 density argument  implies the
 existence of the wave operator.
\begin{lemma}\label{lemma:reduct-proof-theor}
Let $\mu,M,Q_{\mathrm{f}},\delta'$ be as in Lemma~\ref{lemma:10.5.31.23.24}, and 
$u\in \chi_{[\mu,M]}(r^2){\mathcal H}_{\mathrm{aux}}\cap C^\infty(M)$.
Then for any $\varepsilon>0$
there exists $t_0>0$ such that for any $t,t'\geq t_0$ and $v\in C_{\rm c}^\infty(M)$
\begin{align*}
|\inp{v,\e^{\mathrm{i}tH}Q_{\mathrm{f}}(t)U(t)u}
-\inp{v,\e^{\mathrm{i}t'H}Q_{\mathrm{f}}(t')U(t')u}|\leq \varepsilon\|v\|.
\end{align*} 
In particular, $\e^{\mathrm{i}tH}Q_{\mathrm{f}}(t)U(t)u$ is a Cauchy sequence as $t\to\infty$.
\end{lemma}
\proof
Let $\varepsilon>0$.
For any $t\geq t'\geq 1$ and $v\in C_{\rm c}^\infty(M)$ we compute,
using Lemma~\ref{lemma:10.5.31.23.24} \ref{item:10.5.31.23.26} and \ref{item:10.5.31.23.27}  
and the Schwarz inequality,
\begin{align*}
&{}|\inp{v,\e^{\mathrm{i}tH}Q_{\mathrm{f}}(t)U(t)u}
-\inp{v,\e^{\mathrm{i}t'H}Q_{\mathrm{f}}(t')U(t')u}|\\
&{}=\Bigl|\int_{t'}^t\{\inp{v,\e^{\mathrm{i}sH}\mathrm{D}_{G(s)}Q_{\mathrm{f}}(s)U(s)u}
+\mathrm{i}\inp{v,\e^{\mathrm{i}sH}(W(s)+\alpha(s)+V)Q_{\mathrm{f}}(s)U(s)u}\}\,\mathrm{d}s\Bigr|\\
&{}\le \Bigl( \int_{t'}^t \inp{v,\e^{\mathrm{i}sH}(\mathrm{D}_{G(s)}Q_{\mathrm{f}}(s)-R(s))\e^{-\mathrm{i}sH}v}\,\mathrm{d}s\Bigr)^{1/2}\\
&\phantom{{}\le{}}{}\times
\Bigl( \int_{t'}^t \inp{u,U(s)^*(\mathrm{D}_{G(s)}Q_{\mathrm{f}}(s)-R(s))U(s)u}\,\mathrm{d}s\Bigr)^{1/2}
+C \|v\| \|u\| \int_{t'}^ts^{-1-\delta'}\,\mathrm{d}s.
\end{align*}
By Lemma~\ref{lemma:10.5.31.23.24} \ref{item:10.5.31.23.27}  
\begin{align*}
\inp{v,\e^{\mathrm{i}sH}(\mathrm{D}_{G(s)}Q_{\mathrm{f}}(s)-R(s))\e^{-\mathrm{i}sH}v}
=
\tfrac{\mathrm{d}}{\mathrm{d}s}\inp{v,\e^{\mathrm{i}sH}Q_{\mathrm{f}}(s)\e^{-\mathrm{i}sH}v}+O(s^{-1-\delta'})\|v\|^2,
\end{align*}
so that 
\begin{align*}
\Bigl( \int_{t'}^t \inp{v,\e^{\mathrm{i}sH}(\mathrm{D}_{G(s)}Q_{\mathrm{f}}(s)-R(s))\e^{-\mathrm{i}sH}v}\,\mathrm{d}s\Bigr)^{1/2}
\le C\|v\|.
\end{align*}
Similarly, we have
\begin{align*}
\Bigl( \int_{t'}^t  \inp{u,U(s)^*(\mathrm{D}_{G(s)}Q_{\mathrm{f}}(s)-R(s))U(s)u}\,\mathrm{d}s\Bigr)^{1/2}\le C\|u\|,
\end{align*}
which in particular implies that 
$ (\inp{u,U(s)^*(\mathrm{D}_{G(s)}Q_{\mathrm{f}}(s)-R(s))U(s)u}\ge 0$ is integrable.
Hence we obtain
\begin{align*}
&{}|\inp{v,\e^{\mathrm{i}tH}Q_{\mathrm{f}}(t)U(t)u}-\inp{v,\e^{\mathrm{i}t'H}Q_{\mathrm{f}}(t')U(t')u}|\\
&{}\le C 
\|v\|
\Bigl( \int_{t'}^t \inp{u,U(s)^*
(\mathrm{D}_{G(s)}Q_{\mathrm{f}}(s)-R(s))U(s)u}\,\mathrm{d}s\Bigr)^{1/2}
+C \|v\| \|u\| \int_{t'}^ts^{-1-\delta'}\,\mathrm{d}s.
\end{align*}
Since the integrands in the right-hand side both are integrable,
if we let $t_0>0$ be large enough, we have for $t,t'\geq t_0$
\begin{align*}
|\inp{v,\e^{\mathrm{i}tH}Q_{\mathrm{f}}(t)U(t)u}-\inp{v,\e^{\mathrm{i}t'H}Q_{\mathrm{f}}(t')U(t')u}|\leq
\varepsilon \|v\|.
\end{align*}
Thus the lemma follows.
\endproof
For the existence of the limit $\widetilde \Omega_+$  the following
lemma is sufficient. We omit the proof of the lemma.
\begin{lemma}\label{lemma:reduct-proof-theorbb}
Let $E,e,Q_{\mathrm{p}},\delta'$ be as in
Lemma~\ref{lemma:10.3.23.17.20} and $u\in
\chi_{[E-e,E+e]}(H) C_{\rm c}^\infty(M)$. 
Then for any $\varepsilon>0$
there exists $t_0>0$ such that for any $t,t'\geq t_0$ and $v\in C_{\rm c}^\infty(M)$
\begin{align*}
|\inp{v,U(t)^*Q_{\mathrm{p}}(t)\e^{-\mathrm{i}tH}u}
-\inp{v,U(t')^*Q_{\mathrm{p}}(t')\e^{-\mathrm{i}t'H}u}|\leq \varepsilon\|v\|.
\end{align*} 
In particular, $U(t)^*Q_{\mathrm{p}}(t)\e^{-\mathrm{i}tH}u$ is a
Cauchy sequence as $t\to\infty$.
\end{lemma}

\proof[Proof of \eqref{eq:7}, \eqref{eq:2} and Corollary \ref{cor:asymwb}] It suffices to show the  identity
\begin{equation}
  \label{eq:56}
  H\Omega_+=\Omega_+M_f;\;f(x):=2^{-1}r(x)^2.
\end{equation}  Note that the operator $M_f$ has
purely continuous spectrum, given by $[0,\infty)$, so indeed it is a
consequence of \eqref{eq:56} that $\Ran\,\Omega_+\subseteq \vH_{\mathrm{c}}(H)$.
Note that we also have $\Ran\,\widetilde{\Omega}_+\subseteq{\mathcal H}_{\aux}$.
In fact, cf.   the RAGE theorem \cite{RS}, for any $u\in {\mathcal H}_{\mathrm{c}}(H)$
\begin{align*}
\lim_{t\to+\infty}\tfrac{1}{T}\int_1^T\|(P_{\mathrm{aux}})^{\perp}\mathrm{e}^{-\mathrm{i}tH}u\|^2\,\mathrm{d}t=0,
\end{align*}
so that in particular for some sequence $t_n$ 
\begin{align*}
t_n\to \infty\mand (P_{\mathrm{aux}})^{\perp}\mathrm{e}^{-\mathrm{i}t_nH}u\to 0
\end{align*}
as $n\to\infty$.
Then, since $U(t)$ is unitary on $({\mathcal H}_{\aux})^\perp$ and 
$\slim_{t\to+\infty}U(t)^*\mathrm{e}^{-\mathrm{i}tH}P_{\mathrm{c}}$ exists,
we can conclude
\begin{align*}
\slim_{t\to\infty}(P_{\mathrm{aux}})^{\perp}U(t)^*\mathrm{e}^{-\mathrm{i}tH}P_{\mathrm{c}}=0.
\end{align*}
This implies the claim.
Therefore  \eqref{eq:7} follows. Note also that given
\eqref{eq:56} the statements \eqref{eq:2} and Corollary
\ref{cor:asymwb} are immediate consequences of Theorem
\ref{thm:10.5.31.23.23}.

For \eqref{eq:56} we compute for all $s\in \R$
\begin{align*}
 \Omega_+u&=\lim_{t\to \infty} \e^{\i (t+s)H}\mathrm{e}^{\mathrm{i}K(t+s,\cdot)}\mathrm{e}^{-\mathrm{i}\frac{\ln
    (t+s)}{2}A}P_{\mathrm{aux}}u\\&=\e^{\i sH}\lim_{t\to \infty} \e^{\i
  tH}\mathrm{e}^{\mathrm{i}K(t,\cdot)}\mathrm{e}^{\mathrm{i}\tfrac{r^2}{2t^2}\parb{
  \tfrac{t^2}{t+s}-t}}\mathrm{e}^{-\mathrm{i}\frac{\ln
    t}{2}A}P_{\mathrm{aux}}u
\\&=\e^{\i sH}\lim_{t\to \infty} \e^{\i
  tH}\mathrm{e}^{\mathrm{i}K(t,\cdot)}\mathrm{e}^{-\mathrm{i}\frac{\ln
    t}{2}A}\mathrm{e}^{\mathrm{i}f(\cdot)\parb{
  \tfrac{t^2}{t+s}-t}}P_{\mathrm{aux}}u\\
&=\e^{\i sH}\lim_{t\to \infty} \e^{\i tH}U(t)P_{\mathrm{aux}}\e^{-\i sf(\cdot)}u\\&=\e^{\i sH} \Omega_+\e^{-\i sf(\cdot)}u.
\end{align*} 

  Whence, cf.   Stone's theorem \cite{RS},  $H\Omega_+\supseteq\Omega_+M_f$, and therefore also
\eqref{eq:56} holds.
\endproof
 
We end this subsection by also proving  Corollary \ref{cor:asymw}.
\proof[Proof of Corollary \ref{cor:asymw}]
We note $\mathop{\mathrm{Ran}}\Omega_+\subseteq \vH_{\mathrm{c}}(H)$.
Then we compute for all $\phi\in C_{\mathrm{c}}(M)$  and $v\in \vH_{\mathrm{c}}(H)$
\begin{align*}
  \lim_{t\to
   \infty}\e^{\i tH}\phi(\omega(t,\cdot ))\e^{-\i tH}v=\lim_{t\to
   \infty}\e^{\i tH}U(t)\phi(M_x)U(t)^*\e^{-\i tH}v=\Omega_+\phi(M_x)\Omega_+^*v,
\end{align*} showing the existence of $\phi(\omega^+_\infty)=\s -\lim_{t\to
   \infty}\e^{\i tH}\phi(\omega(t,\cdot ))\e^{-\i tH}P_{\mathrm{c}}$ and the
first identity of \eqref{eq:2bb}. 

The existence of the operator $R:=r(\omega^+_\infty)$ follows from the
fact that the mapping $C_{\mathrm{c}}(\R)\ni\psi\to \phi(\omega^+_\infty)\in \vB
(\vH_{\mathrm{c}}(H))$, $\phi:=\psi\circ r$, is a non-degenerate
$*$-representation and spectral theory \cite{RS}. Clearly $R\geq 0$.

We take (using here notation from the next subsection) $\psi_N(s)=\tfrac{s^2}{2}\chi_{-,0,N}(s)$,
$N\in\N$, and take the $N$-limit in the identity $\psi_N
(R)\Omega_+=\Omega_+M_{\psi_N\circ r}$. This leads to
$\tfrac{R^2}{2}\Omega_+\supseteq\Omega_+M_{r^2/2}$,
and whence in combination with \eqref{eq:2}, the
second  identity of \eqref{eq:2bb}. In particular the kernel of $R$ is zero. 
\endproof

\subsection{Localization operators in explicit form}\label{sec:10.6.8.4.0}
The rest of the paper  concerns the proofs of Lemmas~\ref{lemma:10.5.31.23.24} and \ref{lemma:10.3.23.17.20}.
Since the proofs are fairly long, here we first give the explicit
forms of $Q_{\mathrm{f}}$ and $Q_{\mathrm{p}}$. We also collect here
some other 
(related) constructions.

We denote by 
$\chi_{a,b,c,d}\in C^\infty(\mathbb{R})$, $-\infty<a<b<c<d<\infty$, 
a smooth cutoff function such that
\begin{align*}
0\le \chi_{a,b,c,d}\le 1, && \chi_{a,b,c,d}=1 \mbox{ in a nbh.\ of }[b,c], && 
\chi_{a,b,c,d}=0\mbox{ in a nbh.\ of  } \mathbb{R}\setminus (a,d),
\end{align*}
and that
\begin{align*}
\chi_{a,b,c,d}'\ge 0 \mbox{ on } [a,b], &&\chi_{a,b,c,d}'\le 0 \mbox{ on } [c,d],&&
\chi_{a,b,c,d}^{1/2}, |\chi_{a,b,c,d}'|^{1/2}\in C^\infty(\mathbb{R}).
\end{align*} 
We also assume that the family of these cutoff functions satisfies
\begin{align*}
\chi_{a,b,c,d}+\chi_{c,d,e,f}=\chi_{a,b,e,f}, &&
\|\chi_{a,b,c,d}^{(n)}\|_{L^\infty(\mathbb{R})}\le \|\chi_{0,1,2,3}^{(n)}\|_{L^\infty(\mathbb{R})}(\min\{b-a,d-c\})^{-n}.
\end{align*} 
We let $\chi_{-,-,c,d}$ and $\chi_{a,b,+,+}$ be functions with 
 similar  properties as above formally given by taking   
$a=b=-\infty$ and $c=d=+\infty$, respectively. We abbreviate
$\chi_{-,c,d}=\chi_{-,-,c,d}$ and $\chi_{a,b,+}=\chi_{a,b,+,+}$. Note
that all the above functions may be constructed from  $\chi_{0,1,+}$ and
$\chi_{-,0,1}$ by a simple
translation and scaling procedure as well as  multiplication.

Then the localization operators $Q_{\mathrm{f}}$ and $Q_{\mathrm{p}}$ are realized as the products
\begin{align}
Q_{\mathrm{f}}(t)&{}=(Q_2(t)Q_1(t))^*Q_2(t)Q_1(t),\label{10.6.4.23.36}\\ 
Q_{\mathrm{p}}(t)&{}=(Q_6(t)Q_5(t)Q_4(t))^*Q_6(t)Q_5(t)Q_4(t),\label{10.6.6.2.52}
\end{align}
 where we use quantities from the list
\begin{align*}
  Q_1(t)&=\chi_{\mu_1,\mu,M,M_1}\big({r^2}/{t^2}\big),&&\\
Q_2(t)&{}=\parb{I+t^{1+\delta_1}W(t)}^{-1/2},&&\\
Q_3&{}=\chi_{E-2e,E-e,E+e,E+2e}(H),&&\\
Q_4(t)&{}=\chi_{-,2E_1,2E_2}\big({r^2}/{t^2}\big),&&\\
Q_5(t)&{}=\chi_{(1+\delta_3)^2 E/2,(1+\delta_2)^2E/2,+}\bigl({r^2}/{t^2}\big),&&\\
Q_6(t)&{}=Q_2(t){}=\parb{I+t^{1+\delta_1}W(t)}^{-1/2}.&&
\end{align*}
The parameters appearing above are chosen as follows:
For given $0<\mu<M<\infty$, if we let $\mu_1,M_1,\delta_1$ be any constants such that 
\begin{align*}
0<\mu_1<\mu<M<M_1<\infty,&&
0<\delta_1<\min (\delta,2\kappa),
\end{align*}
then $Q_{\mathrm{f}}$ satisfies Lemma~\ref{lemma:10.5.31.23.24}.
For given $E\in (0,\infty)$ let 
$E_{*},\delta_{*}$ be any constants such that 
\begin{align*}
E<E_1<E_2,&&
0<\delta_3<\delta_2<\delta_1<\min (\delta,2\kappa),
\end{align*}
and $e>0$ small enough accordingly,  then $Q_{\mathrm{p}}$ satisfies 
Lemma~\ref{lemma:10.3.23.17.20}.

We shall  consider the following modification of $r^2$ and
corresponding quantities. Pick a
real-valued 
$f\in C^\infty(\R_+)$ with $f(s)=1$ for $s<1/2$, $f(s)=s$ for $s>2$ and
$f''\geq0$. 
Define for any  $\epsilon\in (0,1)$  and all $t\geq1$
\begin{align*}
  \tilde r^2&=t^{2-2\epsilon}f(t^{2\epsilon-2}r^2),\\
\tilde K&=\tfrac{\tilde r^2}{2t},\\
\tilde A&=\i [H_0,\tilde r^2]=\tfrac{1}{2}\{(f'(\cdot)\partial_i
r^2)g^{ij}p_j+p_i^* g^{ij}(f'(\cdot)\partial_j r^2)\},\\
\tilde G&=\tfrac{1}{2}p_i^*g^{ij}p_j-\tfrac{1}{2}(p_i-\partial_i \tilde K)^*g^{ij}(p_j-\partial_j  \tilde K).
\end{align*} 

The latter constructions will be used in Subsection \ref{Preliminary
  localization for  perturbed dynamics} to prove the following
localization for $e>0$ chosen sufficiently small
\begin{subequations}
\begin{align}\label{eq:41}
&\slim_{t\to\infty}(I-Q_3Q_4Q_5^2Q_4Q_3)\e^{-\mathrm{i}tH}\chi_{[E-e,E+e]}(H)=0,
\end{align}
\begin{align}
\begin{split}
&\mForall u\in \chi_{[E-e,E+e]}(H)\vH:\\&\qquad-\int_1^\infty
\inp{\e^{-\mathrm{i}tH}u,(\chi^2_{-,2E_1,2E_2})'\bigl({r^2}/{t^2}\big)\e^{-\mathrm{i}tH}u}t^{-1}\d
t<\infty,\\  
\end{split}
\label{eq:12d}
\end{align}
\begin{align}
\begin{split}
&\mForall u\in \chi_{[E-e,E+e]}(H)\vH:\\&\qquad\int_1^\infty
\inp{\e^{-\mathrm{i}tH}u,(\chi^2_{(1+\delta_3)^2
    E/2,(1+\delta_2)^2E/2,+})'\bigl({r^2}/{t^2}\big)\e^{-\mathrm{i}tH}u}t^{-1}\d
t<\infty.
\end{split}
\label{eq:12}
 \end{align}
\end{subequations} As the reader will see, given
  \eqref{eq:41}--\eqref{eq:12}, the proofs of Lemmas
  \ref{lemma:10.5.31.23.24} and \ref{lemma:10.3.23.17.20} are very similar.

Let $T$ be a self-adjoint operator on a complex  Hilbert space $\vH$ and $\chi\in C^\infty_{\mathrm{c}}(\mathbb{R})$.
 We can choose an almost analytic extension $\tilde{\chi}\in
C^\infty_{\mathrm{c}}(\mathbb{C})$, 
i.e.
\begin{align*}
\tilde{\chi}(x)=\chi(x)\mfor x\in \mathbb{R},&&
|\bar{\partial}\tilde{\chi}(z)|\le C_k|\mathop{\mathrm{Im}}z|^k;\;k\in\N.
\end{align*}
Then the Helffer-Sj\"ostrand representation formula reads
\begin{equation}\label{eq:26}
\chi(T)=\int_{\mathbb{C}}(T-z)^{-1}\,\mathrm{d}\mu(z);\;\mathrm{d}\mu(z)=-\frac{1}{2\pi \mathrm{i}}\bar{\partial}\tilde{\chi}(z)\mathrm{d}z\mathrm{d}\bar{z}.
\end{equation}
If   $S$ is  another  operator on  $\vH$ we are thus lead to the formula
\begin{align}
[S,\chi(T)]=\int_{\mathbb{C}}(T-z)^{-1}
[T,S](T-z)^{-1}\,\mathrm{d}\mu(z).
\label{10.6.3.18.5}
\end{align}
Another well-known
 representation formula  for $T$ strictly positive  reads:
\begin{equation}\label{eq:26oo}
T^{-1/2} = \pi^{-1}\int_0^\infty s^{-1/2} (T+s)^{-1}\d s.
\end{equation}

\section{Verification of properties  of localization operators}\label{sec:Verification of properties  of localization operators}
\subsection{Commutator computations}\label{sec:mourre-estimate}
We compute  several commutators needed later.
We recall that two tensors are denoted by the same symbol if they are related by
the identification $TM\cong T^*M$ through the metric tensor,
and distinguish them by  superscripts and subscripts, for example,
\begin{align*}
((\nabla^2 r^2)^{ij})=(g^{ik}g^{jl}(\nabla^2 r^2)_{kl})\in \Gamma(TM\otimes TM).
\end{align*}
We recall from \cite [Lemma 2.5]{Do} (this formula can
  be proved by a straightforward, although  somewhat tedious,
  computation using the  compatibility condition \eqref{09.11.28.20.51}).
\begin{lemma}\label{prop:09.12.8.11.10} Let $\phi\in
  C^\infty(M)$ be given, and define
  \begin{equation*}
    A_{\phi}=\i[H_0,\phi]=\tfrac{1}{2}\{(\partial_i \phi)g^{ij}p_j+p_i^* g^{ij}(\partial_j \phi)\}.
  \end{equation*}
Then,  as an operator on $C_{\rm c}^\infty(M)$,
\begin{align*}
\mathrm{i}[H_0, A_{\phi}]
=p_i^*(\nabla^2\phi)^{ij}p_j
-\tfrac{1}{4}\triangle^2 \phi.
\end{align*}
\end{lemma}
Let $A$ be the self-adjoint operator defined by (\ref{eq:6}), i.e. we
take $\phi=r^2$ above. From
Lemma  \ref{prop:09.12.8.11.10} we thus obtain
\begin{corollary}\label{cor:10.10.13.15.00}
As a quadratic form on $C_{\mathrm{c}}^\infty(M)$,
  \begin{subequations}
\begin{align}\label{eq:4a}
\mathrm{i}[H,A]
&{}=p_i^*(\nabla^2r^2)^{ij}p_j+
\gamma_ig^{ij}\partial_j+\partial_i^*g^{ij}\gamma_j+\gamma_0;\\
\gamma_i&{}=(\partial_ir^2)V+\tfrac{1}{4}(\partial_i\triangle r^2),
\label{eq:5a}\\
\gamma_0&{}=(\triangle r^2)V.\label{eq:6a}
\end{align}
  \end{subequations}
In particular, for any $\varepsilon>0$ there exists 
$\gamma_\varepsilon=\gamma_\varepsilon(x)=O(r^{-\min\{2\eta,1+2\kappa\}})$ 
such that 
\begin{align}\label{eq:55}
\begin{split}
\mathrm{i}[H,A]
&\ge p_i^*\{(\nabla^2r^2)^{ij}-\varepsilon g^{ij}\}p_j+\gamma_\varepsilon\\
&\ge 2(1+\delta-\varepsilon)H_0-CH_{r_1}+\gamma_\varepsilon,
\end{split}
\end{align}
where $H_{r_1}=\tfrac{1}{2}p_i^*\chi_{-,r_0,r_1}(r)g^{ij}p_j$, $r_1>r_0$.
\end{corollary}
\proof
The equations (\ref{eq:4a})--(\ref{eq:6a}) follow from Lemma~\ref{prop:09.12.8.11.10}
and
\begin{align*}
\mathrm{i}[V,A]=
V(\partial_jr^2)g^{ij}\partial_i+\partial_i^*g^{ij}(\partial_jr^2)V+(\triangle r^2)V.
\end{align*}
If we use (\ref{eq:3}) and 
\begin{align*}
\gamma_ig^{ij}\partial_j+\partial_i^*g^{ij}\gamma_j+\gamma_0
\ge -\varepsilon \partial_i^*g^{ij}\partial_j +\gamma_\varepsilon;\;\gamma_\varepsilon:=- \varepsilon^{-1}g^{ij}\gamma_i\gamma_j+\gamma_0,
\end{align*}
then the latter assertion of the corollary follows. Note 
that indeed since 
$|\partial_r \triangle r^2| \leq C\inp{r}^{-1/2-\kappa}$ we obtain by 
integrating  in $r$  that $\triangle r^2 =O(r^{1/2-\kappa})$.
\endproof
\begin{corollary}\label{cor:09.11.29.16.7}
As a quadratic form on $C_{\mathrm{c}}^\infty(M)$,
\begin{align*}
\mathrm{D}_{H_0}W
={}&-\tfrac{1}{2t}(p_i-\partial_i K)^*
(\nabla^2r^2)^{ij}(p_j-\partial_jK)
+\tilde{\gamma}_i^*g^{ij}(p_j-\partial_j K)
+(p_i-\partial_i K)^*g^{ij}\tilde{\gamma}_j;\\
\tilde{\gamma}_i={}&\tfrac{\mathrm{i}}{8t}(\partial_i\triangle r^2)
-\tfrac{1}{2}(\partial_i\alpha).
\end{align*}
\end{corollary}
\proof
We have  
\begin{align}
\mathrm{D}_{H_0}W=\tfrac{\mathrm{d}}{\mathrm{d}t}W+\mathrm{i}[H_0,W].\label{eq:10.6.4.12.24}
\end{align}
For the first term substitute
\begin{align*}
W=H_0-\tfrac{1}{2}(\partial_iK)g^{ij}p_j-\tfrac{1}{2}p_i^*g^{ij}(\partial_jK)+\tfrac{1}{2}g^{ij}(\partial_iK)(\partial_jK),
\end{align*}
and then we obtain
\begin{align}
\begin{split}
&\tfrac{\mathrm{d}}{\mathrm{d}t}W
=\tfrac{1}{4}(\partial_ig^{kl}(\partial_kK)(\partial_l K))g^{ij}p_j
+\tfrac{1}{4}p_i^*g^{ij}(\partial_jg^{kl}(\partial_kK)(\partial_l K))\\
&-\tfrac{1}{2}g^{ij}(\partial_iK)(\partial_j g^{kl}(\partial_kK)(\partial_lK))
-\tfrac{1}{2}(\partial_i\alpha)g^{ij}p_j
-\tfrac{1}{2}p_i^*g^{ij}(\partial_j\alpha)+g^{ij}(\partial_iK)(\partial_j\alpha).
\end{split}
\label{eq:11.5.5.5.10}
\end{align}
For the second term of (\ref{eq:10.6.4.12.24}) we substitute
\begin{align*}
W=H_0-\tfrac{1}{2t}A+\tfrac{1}{2}g^{ij}(\partial_iK)(\partial_jK),
\end{align*}
and then by Lemma~\ref{prop:09.12.8.11.10}
\begin{align}
\begin{split}
\mathrm{i}[H_0,W]
=&-\tfrac{1}{2t}p_i^*(\nabla^2r^2)^{ij}p_j
-\tfrac{\mathrm{i}}{8t}(\partial_i\triangle r^2)g^{ij}p_j
+\tfrac{\mathrm{i}}{8t}p_i^*g^{ij}(\partial_j\triangle r^2)\\
&+\tfrac{1}{4}(\partial_i g^{kl}(\partial_kK)(\partial_lK))g^{ij}p_j
+\tfrac{1}{4}p_i^*g^{ij}(\partial_jg^{kl}(\partial_kK)(\partial_lK)).
\end{split}
\label{eq:11.5.5.5.11}
\end{align}
Noting the equation, cf. \eqref{09.11.29.3.18},
\begin{align}
(\partial_j g^{kl}(\partial_kK)(\partial_lK))
=2g^{kl}(\nabla K)^2_{jk}(\partial_lK),
\label{eq:11.5.6.8.6}
\end{align}
we obtain the assertion from  (\ref{eq:10.6.4.12.24})--(\ref{eq:11.5.5.5.11}).
\endproof

Introduce the ``radial momentum''  (the name of this operator is
justified by its 
action  on functions supported in $E$)
\begin{equation}\label{eq:29}
  p_r=(\partial_kr)g^{kl}p_l.
\end{equation} 
\begin{lemma} 
  \label{lem:comm-comp}  For any  real-valued 
    $\chi\in C^\infty(\R)$  with $\chi'\in C^\infty_{\mathrm{c}}(\R_+)$
    define
  $p_\chi=\chi(r)p_r$. Then as quadratic forms on
  $\vD \parb{H_0^{1/2}}$
  \begin{subequations}
\begin{align}
  p_\chi^*&=p_\chi-\i(\chi(r)\triangle r+\chi'(r)),\label{eq:4}\\
 p_\chi p_\chi^*&=p_\chi^*p_\chi +\tilde\chi;\;\tilde\chi=\tilde\chi(x):=
 -\chi\parb{\partial_r(\chi(r)\triangle r+\chi'(r))},\label{eq:5}\\
p_\chi^*p_\chi &\leq 2\sup \chi^2 \,H_0,\label{eq:15}\\
p_\chi p_\chi^* &\leq 2\sup \chi^2 \,H_0 +\sup \tilde \chi.\label{eq:16}
\end{align}
  \end{subequations}
\end{lemma}\begin{proof}
 Compute $p_\chi^*=p_\chi+\i\parb{\partial_i^*g^{ij}(\partial_j
   r)\chi(r)}=p_\chi-\i(\chi(r)\triangle r+\chi'(r))$ yielding \eqref{eq:4}. We obtain
 \eqref{eq:5} from \eqref{eq:4} by inserting and commuting  through. The  
 estimate  \eqref{eq:15} follows from the Cauchy Schwarz inequality
 and the fact that $|\nabla r|\le 1$. The 
 estimate \eqref{eq:16} follows from \eqref{eq:5} and
 \eqref{eq:15}.
\end{proof}

In the proof of Lemma \ref{lemma:key_prel-local-pert}  we need the
following technical result which involves the construction $\tilde G$
of Subsection \ref{sec:10.6.8.4.0} given in terms of any $\epsilon\in (0,1)$.
\begin{lemma}
  \label{lemma:comm-comptilkde} There exists
$\epsilon'=\epsilon'(\epsilon, \kappa,\eta)>0$  such that as a quadratic form on
  $C_{\mathrm{c}}^\infty(M)$ 
\begin{align*}
\mathrm{D}_{H}\tilde G
&\geq \tfrac{1}{2t}(p_i-\partial_i K)^*
f'(\cdot)(\nabla^2r^2)^{ij}(p_j-\partial_jK) -Ct^{-\epsilon'-1}H+O(t^{-\epsilon'-1}).
\end{align*}
\end{lemma}
\begin{proof}
  We proceed by computing,  mimicking the proof of Corollaries 
  \ref{cor:10.10.13.15.00} and \ref{cor:09.11.29.16.7},
\begin{align*}
\mathrm{D}_H\tilde G&=\tfrac{\mathrm{d}}{\mathrm{d}t}\tilde
G+\tfrac{\mathrm{i}}{2t}[H,\tilde A]
-\tfrac{\mathrm{i}}{2}[H_0 ,g^{ij}(\partial_i \tilde
K)(\partial_j\tilde K)].
\end{align*}
By (\ref{eq:11.5.6.8.6}) and $(\partial_i\tilde{K})=f'(\cdot)(\partial_iK)$
\begin{align*}
(\partial_t\partial_i\tilde{K})
&=-f'(\cdot)g^{kl}(\nabla^2K)_{ik}(\partial_lK)
+f'(\cdot)(\partial_i \alpha)
+(2\epsilon-2)t^{2\epsilon-3}r^2f''(\cdot)(\partial_iK),
\end{align*}
so that we obtain
\begin{align*}
\tfrac{\mathrm{d}}{\mathrm{d}t}\tilde G
&=-\tfrac{1}{2}(\partial_iK)f'(\cdot)(\nabla^2K)^{ij}(p_j-\partial_j\tilde{K})\\
&\phantom{{}={}}+\tfrac{1}{2}\bigl\{f'(\cdot)(\partial_i \alpha)+(2\epsilon-2)t^{2\epsilon-3}r^2f''(\cdot)(\partial_iK)\bigr\}g^{ij}(p_j-\partial_j\tilde{K})+\mathrm{h.c.}\\
&=-\tfrac{1}{2}(\partial_iK)f'(\cdot)(\nabla^2K)^{ij}(p_j-\partial_jK)\\
&\phantom{{}={}}
-\tfrac{1}{2}f'(\cdot)(\partial_i\alpha)g^{ij}p_j+O(t^{-\epsilon-1})p_r+O(t^{-2\epsilon-1})+\mathrm{h.c.}
\end{align*}
Upon replacing $r^2$ by $\tilde{K}$ in Corollary~\ref{cor:10.10.13.15.00}, we have
\begin{align*}
\tfrac{\mathrm{i}}{2t}[H,\tilde{A}]
={}&p_i^*(\nabla^2\tilde{K})^{ij}p_j
-2\mathop{\mathrm{Im}}\bigl\{\bigl((\partial_i\tilde{K})V+\tfrac{1}{4}(\partial_i\triangle \tilde{K})\bigr)g^{ij}p_j\bigr\}
+(\triangle \tilde{K})V\\
\ge{}& p_i^*f'(\cdot)(\nabla^2K)^{ij}p_j
+\mathop{\mathrm{Im}}\bigl\{O(t^{\max\{-(1-\epsilon)\eta-1,\epsilon-2,-(1/2+\kappa)(1-\epsilon)-1\}})p_r\bigr\}\\
&-\tfrac{1}{2}\mathop{\mathrm{Im}}\bigl\{f'(\cdot) (\partial_i\triangle K)g^{ij}p_j\bigr\}
+O(t^{-(1-\epsilon)(1/2+\kappa+\eta)-1}),
\end{align*}
where in the last step we used the inequality for matrices:
\begin{align*}
 (\nabla^2\tilde r^2)_{ij}&
 = f'(\cdot)(\nabla^2r^2)_{ij}
 +t^{2\epsilon-2}f''(\cdot)(\partial_ir^2)(\partial_jr^2)\geq f'(\cdot)(\nabla^2r^2)_{ij}.
\end{align*}
By (\ref{eq:11.5.6.8.6}) and $(\partial_i\tilde{K})=f'(\cdot)(\partial_iK)$ again
\begin{align*}
-\tfrac{\mathrm{i}}{2}[H_0 ,g^{ij}(\partial_i \tilde K)(\partial_j\tilde K)]
&=-\tfrac{1}{4}(\partial_i f'(\cdot)^2g^{kl}(\partial_k K)(\partial_l K))g^{ij}p_j+\mathrm{h.c.}\\
&=-\tfrac{1}{2} (\partial_i K)f'(\cdot)^2(\nabla^2 K)^{ij}p_j
+O(t^{-\epsilon-1})p_r+\mathrm{h.c.}\\
&=-\tfrac{1}{2}(\partial_i K)f'(\cdot)(\nabla^2 K)^{ij}p_j
+O(t^{-\epsilon-1})p_r+\mathrm{h.c.},
\end{align*}
where we used that for all large $t$ the function $f'(\cdot)$ is supported
in $E$ and whence
\begin{align*}
(\partial_i K)f'(\cdot)^2(\nabla^2K)^{ij}p_j
=\tfrac{r}{t^2} f'(\cdot)^2p_r
=(\partial_i K)(f'(\cdot)\nabla^2K)^{ij}p_j+O(t^{-\epsilon-1})p_r.
\end{align*}
  We  sum and  obtain the assertion.
\end{proof}
\subsection{Further commutator computations}\label{sec:Further
  commutator computations}
In this subsection we collect some further preliminary commutator bounds.

\begin{lemma}
  \label{lemma:furth-comm-comp1} Let $\epsilon>0$ and $0<c<d< a<b$ be
  given. Then uniformly in $t, N\geq 1$
  \begin{equation}
    \label{eq:28}
    \|B\chi_{-,c,d}(r/t)\parb{I+t^{2-2\epsilon}N^{-1}H_0}^{-1}\chi_{a,b,+}(r/t)\|\leq
    C_n\parb{t^{\epsilon }N^{1/2}}^{-n},
  \end{equation} where either $B=B_1=I$ or $B=B_2=t^{1-\epsilon}N^{-1/2}p_r$
  (with $p_r$ given by \eqref{eq:29}) and $n\in \N\cup\{0\}$.
  \end{lemma}
  \begin{proof} Let $T=I+t^{2-2\epsilon}N^{-1}H_0$. For $n=0$ we note
    that
    \begin{equation*}
      [B_2,\chi_{-,c,d}(\cdot)]=-\i t^{-\varepsilon}N^{-1/2}\chi_{-,c,d}'(\cdot)
    \end{equation*} 
which obviously is bounded uniformly in $t\geq1$. 
 Moreover
 \begin{equation}\label{eq:31}
   \|\chi_{-,c,d}(\cdot)B_2T^{-1}\|\leq
   \|B_2T^{-1}\|\leq \sqrt{2}(\sup|\nabla r|)t^{1-\epsilon}N^{-1/2}\|H_0^{1/2}T^{-1}\|\leq 1/\sqrt{2}.
 \end{equation}  This proves \eqref{eq:28} for $n=0$.

For $n\geq 1$ we proceed by induction (using the freedom of using new 
localization functions) first computing
\begin{equation*}
  B\chi_{-,c,d}(\cdot)T^{-1}\chi_{a,b,+}(\cdot)=-\tfrac\i {2} t^{1-2\epsilon}N^{-1}BT^{-1}\parb{p_r^*\chi_{-,c,d}'(\cdot)+\chi_{-,c,d}'(\cdot)p_r}T^{-1}\chi_{a,b,+}(\cdot).
\end{equation*} Now we can freely introduce  a factor $\chi_{-,e,f}(\cdot)$ with
 $d<e<f< a$ in front of the last factor $ T^{-1}$ to the right. By
induction we have 
\begin{equation}
      \label{eq:30}
      \|B\chi_{-,e,f}(\cdot)T^{-1}\chi_{a,b,+}\|\leq C\parb{t^{\epsilon }N^{1/2}}^{-(n-1)} \text{ uniformly in }t,N\geq1.
    \end{equation}
Whence we are left with bounding
\begin{subequations}
  \begin{align}\label{eq:32}
  \|t^{1-2\epsilon}N^{-1}BT^{-1}p_r^*\chi_{-,c,d}'(\cdot)\|&\leq
    Ct^{-\epsilon }N^{-1/2} \text{ uniformly in }t,N\geq1.\\
\|t^{-\epsilon}N^{-1/2}BT^{-1}\chi_{-,c,d}'(\cdot)\|&\leq
Ct^{-\epsilon }N^{-1/2}  \text{ uniformly in }t,N\geq1.\label{eq:33}
\end{align} 
\end{subequations} Clearly \eqref{eq:32} and \eqref{eq:33} in turn follow  from the
following bound:
\begin{equation}
  \label{eq:34}
  \|B_iT^{-1}B_j^*\|\leq 2\;\;i,j\in\{1,2\}.
\end{equation} But as in \eqref{eq:31}
\begin{equation*}
  \|B_1T^{-1}B_j^*\|=\|T^{-1}B_j^*\|=\|B_jT^{-1}\|\leq 1,
\end{equation*} while 
\begin{equation*}
  \|B_2T^{-1}B_j^*\|\leq\|B_2T^{-1/2}\|\times \|T^{-1/2}B_j^*\|\leq 2.
\end{equation*} So indeed \eqref{eq:34} is shown, and the proof of the
lemma is complete.
\end{proof}
\begin{corollary}
  \label{cor:furth-comm-comp} For $\epsilon>0$, $\chi\in
  C^\infty(\R)$ with $\chi'\in C^\infty_{\mathrm{c}}(\R_+)$ and for $B$ given
  as in Lemma \ref{lemma:furth-comm-comp1} we have uniformly  in $t, N\geq 1$
  \begin{equation}
    \label{eq:28g}
    \|B[\chi(r/t),\parb{I+t^{2-2\epsilon}N^{-1}H_0}^{-1}]\|\leq
    Ct^{-\epsilon }.
  \end{equation}
\end{corollary}

\begin{lemma} 
  \label{lem:comm-compss}  For any  real-valued $\chi\in
  C^\infty(\R)$ with $\chi'\in C^\infty_{\mathrm{c}}(\R_+)$ we have uniformly  in $t\geq 1$
  \begin{subequations}
    \begin{align} \label{eq:28ss2}
    \|\inp{H}^{1/2}[Q_3,\chi(r/t)]\inp{H}^{1/2}\|&\leq
    C\sup |\chi'|\,t^{-1},\\
|\inp{H}^{1/2}[Q_3,\parb{\chi'(r/t)p_r+{\rm h.c.}}]\inp{H}^{1/2}\|&\leq
    C_{\chi'}\,t^{-1/2-\kappa}.
\label{eq:28ss22}
  \end{align}
  \end{subequations}
 
\end{lemma}
\begin{proof} We calculate 
\begin{equation}\label{eq:18}
 \i[H,\chi(r/t)]=\tfrac{1}
{2t} \chi'(r/t)p_r+ {\rm h.c}.
\end{equation} Hence in combination with \eqref{10.6.3.18.5} and
\eqref{eq:15}
\begin{align*}
 \|\inp{H}^{1/2}[Q_3,\chi(r/t)]\inp{H}^{1/2}\|&\leq C_1t^{-1}
 \int_{\mathbb{C}}
\|\chi'(r/t)p_r\inp{H}^{-1/2}\|\,\tfrac{\inp{z}^2}{|\Im z|^{2}}|\mathrm{d}\mu(z)|\\
&\leq C_2\sup |\chi'|\,t^{-1} \int_{\mathbb{C}}
\,\tfrac{\inp{z}^2}{|\Im z|^{2}}|\mathrm{d}\mu(z)|\\
&= C_3 \sup |\chi'|\,t^{-1},
\end{align*} showing \eqref{eq:28ss2}.

As for \eqref{eq:28ss22} we rewrite
\begin{equation*}
  2\Re p_{\chi'}=t\{(\partial_i \chi(r/t))g^{ij}p_j+p_i^* g^{ij}(\partial_j \chi(r/t))\},
\end{equation*} and apply Corollary
\ref{cor:10.10.13.15.00}  with the expression $r^2$ replaced  
by $2t\chi(r/t)$. Now we note,
cf. \eqref{eq:49},  that 
\begin{align*}
    t\nabla^2 \chi(\cdot/t)&= t^{-1}\chi''(\cdot/t)\d r\otimes
    \d r+\chi'(\cdot/t)\nabla^2  r,\\
0\leq \nabla^2r&\leq (\triangle r)g\leq C_1
\inp{r}^{-1/2-\kappa}g,\;r\geq r_0,
  \end{align*} leading  to the estimates
\begin{equation*}
    -C_2t^{-1/2-\kappa}g\leq 2t\nabla^2 \chi(\cdot/t)\leq
    C_2t^{-1/2-\kappa}g,\;t\geq 1.
  \end{equation*} Using again \eqref{10.6.3.18.5} this leads to \eqref{eq:28ss22}. 
\end{proof}

In the proof of Lemma \ref{lemma:mini} we need the following technical result.
\begin{lemma} 
  \label{lem:comm-compssC}  For all real-valued $\chi,\hat\chi\in
  C^\infty(\R)$ vanishing for large enough argument and with 
  $\chi',\hat\chi'\in C^\infty_{\mathrm{c}}(\R_+)$
\begin{equation}\label{eq:14}
  \Re\parb{T^*(\Re p_{\hat\chi'(r/t)}-r/t\hat\chi'(r/t))}=T^*\hat\chi'(r/t)(p_r-r/t)+O(t^{-1/2-\kappa}),
\end{equation} where $T=t^{-1}A\chi(r/t)Q_3$.
\end{lemma}
\begin{proof} Introducing
  $\gamma=T^*(\Re p_{\hat\chi'(r/t)}-r/t\hat\chi'(r/t))$ and  noting  that
  $\|T\|$ is uniformly bounded we obtain from \eqref{eq:49} 
  and \eqref{eq:4}  that 
  $\gamma$ has the form of the right hand side of \eqref{eq:14}. It
  remains to show that
  \begin{equation}
    \label{eq:51}
    \gamma^*=\gamma+O(t^{-1/2-\kappa}).
  \end{equation} For that we write
  \begin{equation*}
     \gamma^*=(\Re p_{\hat\chi'})\tfrac At\chi Q_3-(r/t\hat\chi')\tfrac At\chi Q_3,
  \end{equation*} and commute the four factors for each of the two
  terms on the right hand side. Rearranging we then get $\gamma$ plus
  contributions from commutators. The latter are treated using
  repeatedly 
  Corollary  \ref{cor:10.10.13.15.00} and \eqref{eq:18}. The most difficult parts arise from commuting the
  operators $\Re p_{\hat\chi'}$ or $\tfrac At$ through the factors
  of $Q_3$. Here we
  shall only 
  explain how to treat the first term above (the most difficult one). We have 
  \begin{equation}\label{eq:52}
    (\Re p_{\hat\chi'})\tfrac At\chi Q_3=(\Re p_{\hat\chi'})Q_3\tfrac
    At\chi +(\Re p_{\hat\chi'})[\tfrac
    At\chi,Q_3],
  \end{equation}  we represent (introducing here for convenience a suitable
  function $\tilde \chi$ with $\tilde \chi\chi=\chi$)
  \begin{equation*}
    (\Re p_{\hat\chi'})[\tfrac
    At\chi,Q_3] =\int_{\mathbb{C}}(\Re p_{\hat\chi'})(H-z)^{-1}\parb{
\tilde \chi[H,\tfrac
    At]\chi+\tfrac
    At[H,\chi]}(H-z)^{-1}\,\mathrm{d}\mu(z),
  \end{equation*} and then we use Corollary  \ref{cor:10.10.13.15.00}
  and estimate inside the integral. Note the estimates,
  cf. Conditions \ref{cond:conv} and 
  \ref{cond:10.9.2.11.13},  
  \begin{equation}
    \label{eq:53}
    0\leq \nabla^2r^2\leq (\triangle r^2)g\leq C
    \inp{r}^{1/2-\kappa}g,\;r\geq r_0.
  \end{equation}
Since we have the factors $\tilde \chi$ and $\chi$ to the left
and to the right of $[H,\tfrac
    At]$, respectively, the Cauchy Schwarz inequality and these bounds
    lead  to the bound
  $O(t^{-1/2-\kappa})$ of the second  term to the right in \eqref{eq:52}.

Similarly for the first term  in \eqref{eq:52} we write 
\begin{equation}\label{eq:54}
  (\Re p_{\hat\chi'})Q_3\tfrac
    At\chi =Q_3(\Re p_{\hat\chi'}) \tfrac
    At\chi +[\Re p_{\hat\chi'},Q_3]\tfrac
    At\chi.
\end{equation} The second  term of \eqref{eq:54} is
$O(t^{-1/2-\kappa})$ due to \eqref{eq:28ss22}.

Finally for the first term in \eqref{eq:54} we have 
\begin{equation*}
  Q_3(\Re p_{\hat\chi'}) \tfrac
    At\chi =Q_3\chi\tfrac
    At\Re p_{\hat\chi'} +Q_3\chi[\Re p_{\hat\chi'}, \tfrac
    At] +O(t^{-1/2-\kappa}).
\end{equation*} Only the middle term needs examination. We show that
it  is  also  $O(t^{-1/2-\kappa})$ by introducing
$\tilde \chi(s):=\hat\chi'(s)-s\hat\chi''(s)$ and computing  
the adjoint 
\begin{align*}
 &-[\Re p_{\hat\chi'}, \tfrac
    At]\chi Q_3 = \parb{-\tfrac
      2t[p_{\hat\chi'},rp_r]+O(t^{-3/2-\kappa}) }\chi Q_3\\&= \parb{
\tfrac {2\i}{t}p_{\tilde \chi(r/t)}+O(t^{-3/2-\kappa}) }\chi Q_3.
\end{align*} Due to \eqref{eq:15} the norm of this expression is in fact bounded
by $Ct^{-1}$.

\end{proof}

\subsection{Proof of Lemma~\ref{lemma:10.5.31.23.24}}

In this section we let $0<\mu<M<\infty$, and choose $\mu_1,M_1,\delta_1,Q_1,Q_2$ as in Section~\ref{sec:10.6.8.4.0}.

Since $\mathrm{e}^{-\mathrm{i}K(t,\cdot)}U(t)$ is the dilation, the following statement is obvious.
\begin{lemma}\label{lemma:10.6.4.23.37}
For all $u\in \chi_{[\mu,M]}(r^2){\mathcal H}_{\mathrm{aux}}$
\begin{align*}
(1-Q_1)Uu=0.
\end{align*}
\end{lemma}
This lemma can be proved also by the (somewhat formal) equation
\begin{align}
\mathrm{D}_GQ_1P_{\mathrm{aux}}=0.\label{eq:10.6.2.9.7}
\end{align} 
The equation (\ref{eq:10.6.2.9.7}) is obtained by a direct computation.

\begin{lemma}\label{lemma:10.6.4.23.38}
For all  $u\in \chi_{[\mu,M]}(r^2){\mathcal H}_{\mathrm{aux}}$
\begin{align}\label{eq:13}
\lim_{t\to\infty}\inp{Uu,(I-Q_2^2)Uu}=0.
\end{align}
\end{lemma}
\proof Fix $\bar \delta\in (\delta_1,\delta)$ with $\bar \delta<2\kappa$, and
fix $u\in
\chi_{[\mu,M]}(r^2)\vH_{\mathrm{aux}}\cap C^\infty(M)$ (by density \eqref{eq:13}
for 
any such state suffices).
Set for $N\geq1$
\begin{equation*}
  T_N=T_N(t)=I+t^{1+\bar\delta}N^{-1}W(t),
\end{equation*} and note that 
\begin{equation}\label{eq:35}
  \e^{-\i K}T_N\e^{\i K}=I+t^{1+\bar \delta}N^{-1}H_0.
\end{equation}

We need to  show that with $N(t):=t^{\bar\delta-\delta_1}$ 
\begin{align*}
\lim_{t\to\infty}\,\inp{Uu,\parb{I-T_{N(t)}^{-1}}Uu}=0.
\end{align*} 
It suffices to show that there
exists  $R(t)\in{\mathcal B}({\mathcal
    H})$  such that 
\begin{subequations}
\begin{align}\label{eq:37}
&\int_{t_0}^\infty |\inp{U(t)u,R(t)U(t)u}|\,\d t =o(t_0^0)\text{ uniformly in } N\geq1,\\
&\tfrac{\mathrm{d}}{\mathrm{d}t}\inp{U(t)u,(I-T_N^{-1}(t))U(t)u}
\le \inp{U(t)u,R(t)U(t)u}.
\label{eq:10.6.4.13.13}
\end{align}
  \end{subequations}
In fact from \eqref{eq:37} and 
(\ref{eq:10.6.4.13.13}) it follows that 
for any
$\varepsilon>0$ there exists  $t_0\geq1$ such that for all $t\geq t_0$
and $N\geq1$
\begin{align*}
\inp{U(t)u,(I-T_N^{-1}(t))U(t)u}
\le \inp{U(t_0)u,(I-T_N^{-1}(t_0))U(t_0)u}+\varepsilon.
\end{align*}
Then  for all  $N\geq 1$  large enough  we have for all $t\geq t_0$
\begin{align*}
0\le \inp{U(t)u,(I-T_N^{-1})U(t)u}\leq  2\varepsilon.
\end{align*} In particular we can take $N=t^{\bar\delta-\delta_1}$, and
indeed we 
obtain that 
\begin{align*}
0\le \inp{U(t)u,(I-T_{N(t)}^{-1})U(t)u}\leq  2\varepsilon\text{ for
  all sufficiently large } t.
\end{align*}

Hence we only need to prove \eqref{eq:37} and (\ref{eq:10.6.4.13.13}).
We have
\begin{align}
\tfrac{\mathrm{d}}{\mathrm{d}t}U^*(1-T_N^{-1})U
=U^*(-\mathrm{D}_{G}T_N^{-1})U=U^*T_N^{-1}(\mathrm{D}_{H_0}T_N-\i[\alpha,T_N])T_N^{-1}U.
\label{10.6.4.18.57}
\end{align}
By Corollary~\ref{cor:09.11.29.16.7} (in
combination with  an
approximation argument), \eqref{eq:10.2.8.3.55}  and the Cauchy Schwarz inequality
it follows for $r_1>r_0$
\begin{align*}
\mathrm{D}_{H_0}T_N\le 
Ct^{\bar{\delta}}N^{-1}(p_i-\partial_iK)^*
\chi_{-,r_0,r_1}(r)g^{ij}(p_j-\partial_jK)
+\tfrac{2}{\delta-\bar{\delta}}t^{2+\bar{\delta}}N^{-1}|\tilde{\gamma}|^2
\end{align*}
Since $\mathop{\mathrm{supp}}\chi_{-,r_0,r_1}\subset (-\infty,r_1]$,
we obtain 
by using \eqref{eq:35} and Lemma~\ref{lemma:furth-comm-comp1} 
\begin{align}
 \begin{split}
 &  t^{\bar \delta}N^{-1}\|Q_1^*T_N^{-1}(p_i-\partial_iK)^*
 \chi_{-,r_0,r_1}(r)g^{ij}
(p_j-\partial_jK)T_N^{-1}Q_1\|\\
&\leq Ct^{-2}
   \text{ uniformly in } N\geq1.
\end{split}
\label{eq:11.5.5.13.53}
\end{align}
  We claim 
 \begin{equation}
   \label{eq:36}
   t^{2+\bar \delta}N^{-1}\|Q_1^*T_N^{-1}|\tilde{\gamma}|^2T_N^{-1}Q_1\|\leq Ct^{-1+\bar \delta -2\kappa}
   \text{ uniformly in } N\geq1.
 \end{equation} Choose $0<\mu_3<\mu_2<\mu_1$ (with $\mu_1$ as given). 
 Due to \eqref{eq:10.9.2.23.19}
 obviously 
 \begin{subequations}
\begin{equation}\label{eq:39}
   t^{2+\bar \delta}N^{-1}\|\chi_{\mu_3,\mu_2,+}(r^2/t^2)|\tilde{\gamma}|^2\|
   \leq Ct^{-1+\bar\delta-2\kappa}
   \text{ uniformly in } N\geq1,
 \end{equation} which is agreeable with  \eqref{eq:36}. On the other
 hand due to  \eqref{eq:10.9.2.23.19}, \eqref{eq:35}  and Lemma
 \ref{lemma:furth-comm-comp1} (used with $2\epsilon=1-\bar
 \delta$ there) we can estimate 
\begin{equation}
   \label{eq:36b}
   t^{2+\bar \delta}N^{-1}\|Q_1^*T_N^{-1}\chi_{-,\mu_3,\mu_2}(r^2/t^2)|\tilde{\gamma}|^2T_N^{-1}Q_1\|\leq Ct^{-2}
   \text{ uniformly in } N\geq1,
 \end{equation} 
  \end{subequations}which also  agrees with
  \eqref{eq:36}. Using the
  bound $|\d \alpha|\leq C t^{-2}$, cf.  \eqref{eq:11.5.5.13.48},  we obtain for the second term in
  \eqref{10.6.4.18.57}
\begin{equation}
   \label{eq:36kk}
   \|T_N^{-1}\i[\alpha,T_N])T_N^{-1}\|\leq Ct^{-3/2+\bar \delta/2}
   \text{ uniformly in } N\geq1.
 \end{equation} The
  combination of the bounds
  \eqref{eq:39} and \eqref{eq:36b} implies 
 \eqref{eq:36} and therefore, together with (\ref{eq:11.5.5.13.53})
 and \eqref{eq:36kk}, also  \eqref{eq:37} and
 (\ref{eq:10.6.4.13.13}). 
\endproof

\proof[Proof of Lemma~\ref{lemma:10.5.31.23.24}]
Consider the operator $Q_{\mathrm{f}}$ given by (\ref{10.6.4.23.36}).
The property \ref{item:10.5.31.23.25} of
Lemma~\ref{lemma:10.5.31.23.24} for  this operator  follows from Lemmas~\ref{lemma:10.6.4.23.37} and \ref{lemma:10.6.4.23.38}.
By mimicking the proof of Lemma~\ref{lemma:10.6.4.23.38}  we obtain the property
\ref{item:10.5.31.23.26}  for any $\delta'\leq 2\kappa-\delta_1$. 
Finally the property $(W(t)+V+\alpha(t))Q_{\mathrm{f}}(t)\in O_{{\mathcal
    B}({\mathcal H})}(t^{-1-\delta'})$  of \ref{item:10.5.31.23.27}, here possibly
$\delta'>0$ taken smaller,  is proved by first computing
\begin{align*}
&(W+V+\alpha)Q_{\mathrm{f}}\\
&=
t^{-1-\delta_1}Q_1({t^{1+\delta_1}W}Q_2)Q_2Q_1
+[{W},Q_1]Q_2^2Q_1
+t^{-1-\eta}(t^{1+\eta}VQ_1)Q_2^2Q_1+\alpha Q_{\mathrm{f}}.
\end{align*}  The first, third and fourth  terms agree with
\ref{item:10.5.31.23.27}. As for the second  term we
compute
\begin{equation*}
 [{W},Q_1]=-\i
 \chi_{\mu_1,\mu,M,M_1}'\big({r^2}/{t^2}\big)\tfrac{r}{t^2}(\partial_i
 r)g^{ij}(p_j-\partial _jK)- {\rm h.c}.
\end{equation*} Whence we can write
\begin{equation*}
  [{W},Q_1]Q_2=\e ^{\i K}\parb{(-\i t^{-1}\chi(r/t)p_r- {\rm h.c.})(I+t^{1+\delta_1}H_0)^{-1/2}}\e ^{-\i K}
\end{equation*} with $\chi(s)=s\chi_{\mu_1,\mu,M,M_1}'(s^2)$.  
By using this identity, Lemma \ref{lem:comm-comp} and  \eqref{eq:49}
we obtain
\begin{equation*}
\|[{W},Q_1]Q_2\|^2\leq C \parb{t^{-3-\delta_1}+t^{-3-2\kappa}}.  
\end{equation*} 
\endproof

\subsection{Preliminary localization for  perturbed dynamics}\label{Preliminary localization for  perturbed dynamics}
In this subsection we first study  various preliminary localization properties of  the
perturbed 
dynamics. We prove maximal and minimal velocity bounds and in
particular the properties \eqref{eq:41}--~\eqref{eq:12}. Similar
properties were also used in the proof of
Lemma~\ref{lemma:10.5.31.23.24}, cf. Lemma \ref{lemma:10.6.4.23.37}. 
 However the proofs for the perturbed dynamics  are  somewhat
 technical. Since  all we need from this
 subsection for the proof of Lemma \ref{lemma:10.3.23.17.20} is in fact the
 properties 
 \eqref{eq:41}--\eqref{eq:12} the reader might prefer
 to read the next subsection (presenting a proof of  Lemma
 \ref{lemma:10.3.23.17.20} along the lines of the proof of Lemma
 \ref{lemma:10.5.31.23.24}) before coming back to the present
 one. The subsection depends on \cite{Gr, SS} although the
 presentation is self-contained.

Let $E\in (0,\infty)$ and we fix $Q_*$ and the parameters 
$E_*,\delta_*$ as in Section~\ref{sec:10.6.8.4.0}.
The small parameter $e>0$ will be determined in this section. Possibly
we will
retake it smaller each time it appears.
The following type of result is called a {\it Mourre estimate} in the
literature since the appearance of such estimate in the seminal work \cite{Mo}. The
reader should keep in mind though that the commutator in
Corollary~\ref{cor:10.10.13.15.00} does
not conform with the conditions of \cite{Mo} since under our conditions it
might not be bounded relative to $H$ (not even
in the form sense). At this point we remark that Donnelly  \cite{Do} indeed uses
 Mourre theory under his geometric conditions. In fact our conditions do not
conform neither with more recent refinement
of Mourre theory as a method to provide the limiting absorption
principle  \cite{MS,GGM,FMS}. However as the reader will see we are
not going to use this theory, or more generally limiting absorption
bounds, only the following reminiscence.
\begin{lemma}\label{lemma:10.1.23.9.52}
 For 
 $e>0$  sufficiently small and as a form estimate on $C^\infty_{\mathrm{c}}(M)$
\begin{align}\label{eq:8}
\mathrm{i}[H,A]\ge 2(1+\delta_1)E -C(I-Q_3)(H_{r_1}+1)(I-Q_3).
\end{align} 
\end{lemma}
\proof Fix $\varepsilon>0$ such that $\delta_1<\delta-3\varepsilon$. 
By \eqref{eq:10.2.8.3.55} and \eqref{eq:55}
\begin{align*}
\mathrm{i}[H,A]&
\ge 2(1+\delta-\varepsilon)H_0 -CH_{r_1}+\gamma_{\varepsilon}
=2(1+\delta-\varepsilon)H -CH_{r_1}+o(\langle r\rangle^0)\\
&\geq 2(1+\delta-2\varepsilon)Q_3HQ_3-C_1(\varepsilon)(I-Q_3)(H_{r_1}+1)(I-Q_3)+K;\\ 
K&=Q_3\{ -CH_{r_1}+o(\langle r\rangle^0)\}Q_3.
\end{align*}
Since $K$ is compact and $E\not\in \sigma_{\mathrm{pp}}(H)$,
we can make $\|K\|_{{\mathcal B}({\mathcal H})}$ arbitrary small by 
letting $e>0$ small. Hence if $e>0$ is sufficiently small we obtain 
\begin{align*}
\mathrm{i}[H,A]&\ge 2(1+\delta-3\varepsilon)(E-2e)
-C_2(\varepsilon)(I-Q_3)(H_{r_1}+1)(I-Q_3)\\&\ge 2(1+\delta_1)E-C_2(\varepsilon)(I-Q_3)(H_{r_1}+1)(I-Q_3).
\end{align*} 
\endproof

The following type of result is called a {\it maximal velocity bound} in the
literature. We shall present a somewhat different proof than 
seen in for example  \cite{CHS1, Gr}.  It is more in the spirit of the proof of Lemma \ref{lemma:10.6.4.23.38}.  
\begin{lemma}\label{lemma:10.1.16.14.53}
If $e>0$ is sufficiently small, then
for any $u\in \chi_{[E-e,E+e]}(H){\mathcal H}$
\begin{align*}
\lim_{t\to\infty}\inp{\e^{-\mathrm{i}tH}u,(I-Q_{4})\e^{-\mathrm{i}tH}u}=0.
\end{align*}
\end{lemma}
\proof

\smallskip
\noindent
\textit{Step 1.}
Set $\chi=\chi_{-,N,2N}$. We first prove 
\begin{align}
\lim_{N\to\infty}\limsup_{t\ge 1}{}\inp{\e^{-\mathrm{i}tH}u,\big(I-\chi\bigl(r^2/t^2\big)\big)\e^{-\mathrm{i}tH}u}=0.
\label{10.1.16.13.19}
\end{align}
 For that it suffices to show that there
exist $\delta'>0$ and $R\in{\mathcal B}({\mathcal
    H})$  such that 
\begin{subequations}
\begin{align}\label{eq:37l}
\|R\|&\leq Ct^{-1-\delta'}\text{ uniformly in } N\geq1,\\
\tfrac{\mathrm{d}}{\mathrm{d}t}\inp{\e^{-\mathrm{i}tH}u,\big(I-\chi\bigl(r^2/t^2\big)\big)\e^{-\mathrm{i}tH}u}
&\le \inp{u,Ru}.
\label{eq:10.6.4.13.13l}
\end{align}
\end{subequations}
We calculate, cf. \eqref{eq:3} and (\ref{eq:18}),
\begin{equation}\label{eq:18w}
 \D _H\chi(r^2/t^2)=
 \begin{cases}
   \tfrac{r}
{t^2} \chi'(r^2/t^2)(p_r-r/t)+ {\rm h.c.}\\
\tfrac{r}
{t^2} \chi'(r^2/t^2)(\nabla r)_ig^{ij}(p-\nabla K)_j+ {\rm
  h.c.}
 \end{cases}. 
\end{equation} We will now use the first identity in
 \eqref{eq:18w}.  (For the second identity we use implicitly that $t$ is
 large.) Clearly there is here the  positive term 
 $-2\tfrac {r^2}{t^3} \chi'(r^2/t^2)$ which  for $t$ 
 large is equal to  $2\tfrac{r}
{t^2} \tilde\chi^2(r/t)$ where $\tilde\chi(s)
=\sqrt{-|s|\chi'(s^2)}$.  The remaining term is
 symmetrized as
 \begin{equation}\label{eq:44}
   \tfrac{r}
{t^2} \chi'(r^2/t^2)p_r+ {\rm
  h.c.}=-t^{-1}\tilde\chi(r/t)(p_r+p_r^*)\tilde\chi(r/t).
 \end{equation}
Next we use \eqref{eq:15}  with $\chi=1$
\begin{equation*}
  \|p_rQ_3\|, \|Q_3p_r^*\|\leq C=\sqrt{2\|Q_3H_0Q_3\|},
\end{equation*} yielding the lower bound 
\begin{align*}
  \tfrac{r}
{t^2} \chi'(r^2/t^2)(p_r-r/t)+ {\rm h.c.}\geq&2(\sqrt
N-C)t^{-1}\tilde\chi^2(r/t)\\
& -t^{-1}\tilde\chi(r/t)\parb{p_r(I-Q_3)+(I-Q_3)p_r^*}\tilde\chi(r/t).
\end{align*}
Due to Lemma
\ref{lem:comm-compss} the second term to the right contributes to
\eqref{eq:10.6.4.13.13l} by a term  whose
norm is bounded by 
$CN^{-1}t^{-2}$. Whence for $N\geq C^2$ indeed we obtain  \eqref{eq:37l}
and \eqref{eq:10.6.4.13.13l} with $\delta'=1$.

Due to Step 1  it suffices to show that for any fixed $N$ 
\begin{equation}
\lim_{t\to\infty}\inp{\e^{-\mathrm{i}tH}u,\psi\bigl(r^2/t^2\big)\e^{-\mathrm{i}tH}u}=0,
\label{10.1.15.11.12}
\end{equation}
where 
$\psi=\chi_{2E_1,2E_2,N,2N}$.
For that we need two more steps.

\smallskip
\noindent
\textit{Step 2.} 
We prove  that 
\begin{align}
\int_1^\infty\inp{\e^{-\mathrm{i}tH}u,\psi\bigl(r^2/t^2\big)\e^{-\mathrm{i}tH}u}\,t^{-1}\d t<\infty.
\label{10.1.20.20.08}
\end{align}
Put
\begin{align*}
\chi(s)=\int_{-\infty}^{s} \psi(\beta^2)\,\mathrm{d}\beta,
\end{align*}
and compute, cf. \eqref{eq:18}, \eqref{eq:18w}  and \eqref{eq:44}, 
\begin{align*}
&{}\tfrac{\mathrm{d}}{\mathrm{d}t}
\e^{\mathrm{i}tH}\chi(r/t)\e^{-\mathrm{i}tH}\\
&{}=\tfrac{1}{2t}\e^{\mathrm{i}tH}
\psi\bigl(r^2/t^2\big)^{1/2}\parb{(p_r-r/t)+{\rm h.c.}}
\psi\bigl(r^2/t^2\big)^{1/2}
\e^{-\mathrm{i}tH}\\
&{}\le \tfrac{1}{t}\e^{\mathrm{i}tH}
\psi\bigl(r^2/t^2\big)^{1/2}\parb{\Re p_r-\sqrt{2E_1}}
\psi\bigl(r^2/t^2\big)^{1/2}
\e^{-\mathrm{i}tH}.
\end{align*} Next to treat the contribution from $\Re p_r$ we proceed again as in Step 1 
inserting  factors of $Q_3$, and using  (in the first estimation) that by \eqref{eq:15} 
\begin{equation*}
 \Re p_r\leq \epsilon/2 +\epsilon^{-1}H_0\text{ for all } \epsilon>0,
\end{equation*}
\begin{align*}
  &\tfrac{1}{t}\e^{\mathrm{i}tH}
\psi\bigl(r^2/t^2\big)^{1/2}\Re p_r
\psi\bigl(r^2/t^2\big)^{1/2}
\e^{-\mathrm{i}tH}\\
&=\tfrac{1}{t}\e^{\mathrm{i}tH}
\psi\bigl(r^2/t^2\big)^{1/2}Q_3\Re p_rQ_3
\psi\bigl(r^2/t^2\big)^{1/2}
\e^{-\mathrm{i}tH}+\text{ remainder}\\
&\le \tfrac{1}{t}\e^{\mathrm{i}tH}
\psi\bigl(r^2/t^2\big)^{1/2}Q_3\parb{\epsilon/2+\epsilon^{-1}(H-V)}Q_3
\psi\bigl(r^2/t^2\big)^{1/2}
\e^{-\mathrm{i}tH}+\text{ remainder}\\
&\le \tfrac{1}{t}\e^{\mathrm{i}tH}
\psi\bigl(r^2/t^2\big)^{1/2}Q_3\parb{\epsilon/2+\epsilon^{-1}(E+2e-V)}Q_3
\psi\bigl(r^2/t^2\big)^{1/2}
\e^{-\mathrm{i}tH}+\text{ remainder}\\
&\le \tfrac{1}{t}\parb{\epsilon/2+\epsilon^{-1}(E+2e)}
\e^{\mathrm{i}tH}
\psi\bigl(r^2/t^2\big)
\e^{-\mathrm{i}tH}+\text{ remainder}.
\end{align*}  
 The  remainders are treated
by Lemma
\ref{lem:comm-compss} and Condition \ref{cond:10.6.1.16.24}. They have
norms bounded by 
  $Ct^{-2}$.  Taking $\epsilon=\sqrt{2(E+e)}$ we thus obtain
\begin{align*}
&{}\tfrac{\mathrm{d}}{\mathrm{d}t}
\chi_{[E-e,E+e]}(H)\e^{\mathrm{i}tH}\chi(r/t)\e^{-\mathrm{i}tH}\chi_{[E-e,E+e]}(H)\\
&{}\leq-\tfrac{c}{t}
\chi_{[E-e,E+e]}(H)\e^{\mathrm{i}tH}
\psi\bigl(r^2/t^2\big)
\e^{-\mathrm{i}tH}\chi_{[E-e,E+e]}(H)+O(t^{-2});\\&\quad c=\sqrt{2E_1}-\sqrt{2(E+2e)}.
\end{align*} For $e>0$ small enough the constant $c>0$. Then 
\eqref{10.1.20.20.08} follows by integration and by using that $\chi$
is bounded.

\smallskip
\noindent
\textit{Step 3.} 
We prove (\ref{10.1.15.11.12}).
By \eqref{10.1.20.20.08} there exists a sequence $t_n\to \infty$ such
that
\begin{equation*}
\inp{\e^{-\mathrm{i}t_nH}u,\psi\bigl(r^2/t_n^2\bigr)\e^{-\mathrm{i}t_nH}u}\to 0.
  \end{equation*}
Thus it suffices to show that 
\begin{equation}\label{eq:19}
\int_1^\infty \bigl|\tfrac{\mathrm{d}}{\mathrm{d}t}
\inp{\e^{-\mathrm{i}tH}u,\psi\bigl(r^2/t^2\bigr)\e^{-\mathrm{i}tH}u}\bigr|\,\d t<\infty.
\end{equation}
 But by   calculations and estimations like in Step 2 we obtain that 
\begin{equation}\label{eq:43}
\bigl|\tfrac{\mathrm{d}}{\mathrm{d}t}
\inp{\e^{-\mathrm{i}tH}u,\psi\bigl(r^2/t^2\bigr)\e^{-\mathrm{i}tH}u}\bigr|
\le \tfrac{C}{t}\inp{\e^{-\mathrm{i}tH}u,|\psi'\bigl(r^2/t^2\bigr)|
\e^{-\mathrm{i}tH}u}
+\tfrac{C}{t^2}\|u\|^2.
\end{equation}
 We fix a non-negative function $\tilde \psi\in
 C_{\mathrm{c}}^\infty((E+E_1,\infty))$ with $\tilde \psi \psi'=\psi'$. By using
 Step 2 to this function (instead of $\psi$) we obtain
 \eqref{10.1.20.20.08} with $\psi$ replaced 
 $\tilde \psi$ (possibly by taking $e>0$ smaller).
 Combining this bound with \eqref{eq:43} we obtain \eqref{eq:19}, and 
hence the lemma follows.
\endproof

\begin{corollary}\label{cor:10.6.9.4.24}
For  all  $ u\in \chi_{[E-e,E+e]}(H)\vH$ with $e>0$  taken sufficiently small the bound \eqref{eq:12d} holds.
\end{corollary}

Next we shall show a version of the key phase space propagation
estimate of \cite{Gr, SS} using here the
quantities enlisted before \eqref{eq:41}--\eqref{eq:12}.
\begin{lemma}
  \label{lemma:key_prel-local-pert} For all $u\in
  \chi_{[E-e,E+e]}(H){\mathcal H}$ and $\chi\in
  C_{\mathrm{c}}^\infty(\R_+)$:
  \begin{equation}
    \label{eq:45}
    \int_1^\infty
\inp{\chi(r/t)\e^{-\mathrm{i}tH}u,W^{1/2}\chi(r/t)\e^{-\mathrm{i}tH}u}t^{-1}\d
t<\infty.
  \end{equation}
\end{lemma}
\begin{proof} Let  $u\in
  \chi_e(H){\mathcal H}$. Here and henceforth we  abbreviate
  $\chi_e(H)=\chi_{[E-e,E+e]}(H)$. Let  $\chi\in
  C_{\mathrm{c}}^\infty(\R_+)$ and choose  $\hat\chi\in
  C_{\mathrm{c}}^\infty(\R_+)$ with $\hat\chi=1$ on the support of $\chi$. Possibly by enlarging $E_1<E_2$ we can assume
  that $\hat\chi(r/t)=\hat\chi(r/t)Q_4$. We are going to use the conclusion of
  Corollary \ref{cor:10.6.9.4.24} for this $Q_4$. Note here that such
  ``enlargement'' is doable uniformly in the small
  parameter $e>0$.

\noindent 
\textit{Step 1.} We show that for any $\epsilon\in (0,1)$
\begin{equation}
    \label{eq:45b}
    \int_1^\infty
|\inp{ (p_i-\partial_i K)Q_4\e^{-\mathrm{i}tH}u,
f'(t^{2\epsilon-2}r^2)(\nabla^2r^2)^{ij}(p_j-\partial_jK)Q_4\e^{-\mathrm{i}tH}u}|t^{-1}\d
t<\infty.
  \end{equation}We shall use the family of observables
  $\chi_e(H)Q_4\tilde
  GQ_4\chi_e(H)$. Clearly this family is bounded.  We calculate the
  Heisenberg derivative. For the 
  leading term coming from the derivative 
  of $\tilde G$ we invoke Lemmas \ref{lemma:comm-comptilkde} and
  \ref{lem:comm-compss} yielding 
\begin{align*}
&\chi_e(H)Q_4(\mathrm{D}_{H}\tilde G)Q_4\chi_e(H)\\
&\geq\tfrac{1}{2t}\chi_e(H)Q_4(p_i-\partial_i K)^*
f'(\cdot)(\nabla^2r^2)^{ij}(p_j-\partial_jK)Q_4\chi_e(H) +O(t^{-\epsilon'-1}).
\end{align*} 
  
Combining this estimate with
\begin{equation}
    \label{eq:45bc}
    \int_1^\infty
|\Re \inp{ (\D_HQ_4)\e^{-\mathrm{i}tH}u,\tilde GQ_4\e^{-\mathrm{i}tH}u}|\d
t<\infty,
  \end{equation}  
we obtain \eqref{eq:45b} by integration. In turn \eqref{eq:45bc}
follows using first that
\begin{align*}
  &\chi_e(H)(\D_HQ_4)\tilde GQ_4\chi_e(H)+\text{
    h.c. }\\ &=\chi_e(H)\Re\parb{(\D_HQ_4^2)\tilde G}\chi_e(H)+O(t^{-2}),
\end{align*} and then invoking \eqref{eq:18w} and  Lemma \ref{lem:comm-compss}  to rewrite the first term as
\begin{equation*}
  \chi_e(H)\Re\parb{(\D_HQ_4^2)\tilde G}\chi_e(H)=t^{-1}\chi_e(H)\Re\parb{\tilde\chi'(r/t)B}\chi_e(H)+O(t^{-3/2-\kappa})
\end{equation*} with  $\tilde\chi(s)=\chi_{-,2E_1,2E_2}^2(s^2)$ and
$B=B(t)$  being  uniformly bounded. Using \eqref{eq:28ss2} again we  also
conclude that 
 \begin{equation}
   \label{eq:46}
-\Re\parb{\tilde\chi'(\cdot)B}=
\sqrt{|\tilde\chi'|(\cdot)}\Re(B)\sqrt{|\tilde\chi'|(\cdot)}+O(t^{-1}).
 \end{equation}  Whence  the contribution from the first term
 in \eqref{eq:46} can be   treated by using \eqref{eq:12d}  while the
 contribution from the second term as well as previous error terms clearly are integrable.

\noindent
\textit{Step 2.}  We show \eqref{eq:45}. From \eqref{eq:45b} we can deduce the estimate
\begin{equation}
    \label{eq:45bk}
    \int_1^\infty
\inp{ \chi(r/t)\e^{-\mathrm{i}tH}u,W\chi(r/t)\e^{-\mathrm{i}tH}u}t^{-1}\d
t<\infty.
  \end{equation} This estimate also holds with  $\chi\to \hat\chi$. The proof goes as follows: First we estimate  (using commutation)
\begin{align*}
 &\|\sqrt{W}\chi(r/t)\e^{-\mathrm{i}tH}u\|^2\\
    &=\|\sqrt{W} \chi(r/t)
      Q_4\e^{-\mathrm{i}tH}u\|^2\\
&\leq \inp{ (p_i-\partial_i K) Q_4\e^{-\mathrm{i}tH}u, \chi(r/t)^2
g^{ij}(p_j-\partial_jK) Q_4\e^{-\mathrm{i}tH}u}+Ct^{-1}\|u\|^2\\
&\leq \parb{\sup | \chi|^2}\| f'(t^{2\epsilon-2}r^2)^{1/2}(p-\nabla K) Q_4\e^{-\mathrm{i}tH}u\|^2+t^{-1}C\|u\|^2.
  \end{align*}
Due to 
Condition \ref{cond:conv} and \eqref{eq:45b} we have for a large
$t_0\ge 1$
\begin{align*}
    &\int_1^\infty\| f'(t^{2\epsilon-2}r^2)^{1/2}(p-\nabla K) Q_4\e^{-\mathrm{i}tH}u\|^2t^{-1}\,\d t\\ 
&\leq \int_1^{t_0}\| f'(t^{2\epsilon-2}r^2)^{1/2}(p-\nabla K) Q_4\e^{-\mathrm{i}tH}u\|^2t^{-1}\,\d t\\
&\phantom{{}\leq {}}{}+(1+\delta)^{-1}\int_{t_0}^\infty
\inp{ (p_i-\partial_i K)Q_4\e^{-\mathrm{i}tH}u,
f'(t^{2\epsilon-2}r^2)(\nabla^2r^2)^{ij}(p_j-\partial_jK)Q_4\e^{-\mathrm{i}tH}u}t^{-1}\d t\\
&<\infty,
  \end{align*} showing then \eqref{eq:45bk}.

Consider now the uniformly bounded observables
  (cf. \eqref{eq:48} given below)
  \begin{equation*}
  B^*
  \sqrt{W+t^{-\sigma}}B;\;B:=\chi(r/t)\chi_e(H)\mand \sigma\in(0,2).
  \end{equation*} 
We use \eqref{eq:26oo} to write 
\begin{equation}\label{eq:26ooc}
\sqrt{W+t^{-\sigma}}= \pi^{-1}\int_0^\infty s^{-1/2} (W+t^{-\sigma}+s)^{-1}(W+t^{-\sigma})\d s,
\end{equation} and then in turn calculate
\begin{equation*}
\D_{H_0}\sqrt{W+t^{-\sigma}}= \pi^{-1}\int_0^\infty s^{1/2}
(W+t^{-\sigma}+s)^{-1}\parb{(\D_{H_0}W)-\sigma
  t^{-\sigma-1}}(W+t^{-\sigma}+s)^{-1}\d s.
\end{equation*}
  Next we apply Corollary
  \ref{cor:09.11.29.16.7}  and use the Cauchy Schwarz inequality as in the proof of Lemma
  \ref{lemma:10.6.4.23.38}. This   leads  to the following bound for any $r_1>r_0$ and suitable
  constants $c,C>0$:
\begin{align}
\begin{split}
    &B^*\parbb{\D_{H_0}\sqrt{W+t^{-\sigma}}}B\\
  &\leq
    -\tfrac{c}{t}B^*W^{1/2}B
    +\tfrac{C}{t}\int_0^\infty s^{1/2}B^*(W+t^{-\sigma}+s)^{-1}\\
  &\phantom{{}\le{}} \bigl[
    (p_i-\partial_iK)^* \chi_{-,r_0,r_1}(r)g^{ij}
  (p_j-\partial_jK)+t^2|\tilde{\gamma}|^2\bigr]
   (W+t^{-\sigma}+s)^{-1}B\,\d s.
   \end{split}\label{eq:11.7.3.7.23}
  \end{align}   
Next we note  that (minus) the integral of the first term is
  the quantity that enters in \eqref{eq:45}, so it suffices to show
  that the contribution from the second term   is integrable  as well
  as to show the bounds
  \begin{subequations}
\begin{align}
    \label{eq:45bcd}  &\int_1^\infty\|B^*V\sqrt{W+t^{-\sigma}}B\|\,\d
t<\infty,\\
    &\int_1^\infty
|\inp{ (\D_{H_0}\chi(r/t))\e^{-\mathrm{i}tH}u,\sqrt{W+t^{-\sigma}}\chi(r/t)\e^{-\mathrm{i}tH}u}|\,\d
t<\infty.\label{eq:47}
  \end{align}
  \end{subequations}
As for the contribution from the second term in the square bracket in (\ref{eq:11.7.3.7.23}) we estimate using Lemma \ref{lemma:furth-comm-comp1} with $d>c>0$
chosen to the left of the support of $\chi$, $\epsilon >0$ such that
$2-2\epsilon=\sigma$ and  $N=N(s,t)\geq 1$  such that $Nt^{-\sigma}=t^{-\sigma}+s$ 

\begin{align}
\begin{split}
  &\int_0^\infty
s^{1/2}\|B^*(W+t^{-\sigma}+s)^{-1}t^2|\tilde{\gamma}|^2\chi_{-,c,d}(r/t)(W+t^{-\sigma}+s)^{-1}B\|\,\d s\\
&\leq C_nt^{-\epsilon n}\int_0^\infty
s^{1/2}(t^{-\sigma}+s)^{-2}\,\d
s\\&=\tilde C_nt^{\sigma/2-\epsilon n},
\end{split}\label{eq:11.7.3.7.24}
\end{align}
Choosing $n\in \N$ large enough gives integrability in $t$ of this contribution. Similarly
\begin{align*}
  &\int_0^\infty
s^{1/2}\|B^*(W+t^{-\sigma}+s)^{-1}t^2|\tilde{\gamma}|^2\chi_{c,d,+}(r/t)(W+t^{-\sigma}+s)^{-1}B\|\,\d s\\
&\leq Ct^{-1-2\kappa}\int_0^\infty
s^{1/2}(t^{-\sigma}+s)^{-2}\,\d
s\\&=\tilde Ct^{\sigma/2-1-2\kappa},
\end{align*}  yielding integrability in $t$ of this contribution. 
As for the contribution from the first term in the square bracket in (\ref{eq:11.7.3.7.23}) 
we decompose 
\begin{equation*}
    (p_i-\partial_iK)^* \chi_{-,r_0,r_1}(r)g^{ij}
  (p_j-\partial_jK)=2W\chi_{-,r_0,r_1}(r)+\i (p_i-\partial_iK)^* g^{ij}(\partial_jr)\chi_{-,r_0,r_1}'(r).
\end{equation*} In combination  with the factor $(W+t^{-\sigma}+s)^{-1}$ to
the left  we thus obtain the uniform bound
\begin{align*}
    &(W+t^{-\sigma}+s)^{-1}(p_i-\partial_iK)^* \chi_{-,r_0,r_1}(r)g^{ij}
  (p_j-\partial_jK)\\ &=O_{\vB (\vH)}(1)\chi_{-,r_0,r_1}(r)+O_{\vB (\vH)}\parb{(t^{-\sigma}+s)^{-1/2}}\chi_{-,r_0,r_1}'(r).
\end{align*} Next we note that it suffices to consider integrability
in $t\in [t_0,\infty)$ for any sufficiently large $t_0\geq 1$ (rather
than in $t\in[1,\infty)$). We pick
$t_0$ such that we can freely insert the above factor
$\chi_{-,c,d}(r/t)$ to the right (for example $t_0=\max(r_1/c,1)$). Once this
factor is inserted we invoke again Lemma \ref{lemma:furth-comm-comp1}
with $\epsilon$ and $N$ chosen as above. We thus obtain for any $n\geq 0$
\begin{equation*}
  \|\chi_{-,c,d}(r/t)(W+t^{-\sigma}+s)^{-1}B\|\leq C_n t^{-n}(t^{-\sigma}+s)^{-n/2-1},
\end{equation*} and to  conclude   we need to estimate (for some $n$)
\begin{align}\label{eq:63}
\int_{t_0}^\infty\,\d t\int_0^\infty
  s^{1/2}t^{-n-1+\sigma/2}(t^{-\sigma}+s)^{-n/2-1}\,\d s<\infty.
\end{align} Indeed \eqref{eq:63} is true for any $n>1$.

So we are left with proving \eqref{eq:45bcd} and \eqref{eq:47}. As for
\eqref{eq:45bcd} we note that $\|B^*V\|=O(t^{-1-\eta})$, hence
integrable, and that
\begin{align}
\begin{split}
\|\sqrt{W+t^{-\sigma}}B\|^2&\leq 
\|B^*WB\|+C_1\leq  \|B^*\parb{2H_0+|\nabla
  K|^2}B\|+C_1
\\ &\leq 2
\|B^*HB\|+C_2
\leq C_3,
\end{split}
\label{eq:48}
\end{align} 
where we in the last step used Lemma \ref{lem:comm-compss}.

As for \eqref{eq:47} we write  
\begin{equation*}
  \D_{H_0}\chi(r/t)=t^{-1}\Re\gamma;\;\gamma:=\chi'(r/t)(p_r-r/t).
\end{equation*} We have, cf.  \eqref{eq:49} and \eqref{eq:4},
\begin{equation*}
\Re\gamma=\gamma +O(t^{-1/2-\kappa}).
\end{equation*}

Whence using the Cauchy Schwarz
inequality we can  estimate 
\begin{align*}
  &|\inp{
    (\D_{H_0}\chi(r/t))\e^{-\mathrm{i}tH}u,\sqrt{W+t^{-\sigma}}\chi(r/t)\e^{-\mathrm{i}tH}u}|\\
&\leq  t^{-1}\parbb{\| \gamma\e^{-\mathrm{i}tH}u\|+Ct^{-1/2-\kappa}\|u\|}\|\sqrt{W+t^{-\sigma}}B\e^{-\mathrm{i}tH}u\|\\
&\leq C_1 t^{-1}\parbb{\| \sqrt{W}\hat
  \chi(r/t)\e^{-\mathrm{i}tH}u\|+t^{-1/2-\kappa}}\|\sqrt{W+t^{-\sigma}}B\e^{-\mathrm{i}tH}u\|\\
&\leq C_1 t^{-1}\| \sqrt{W}\hat \chi(r/t)\e^{-\mathrm{i}tH}u\|\times 
\|\sqrt{W+t^{-\sigma}}B\e^{-\mathrm{i}tH}u\|+C_2t^{-3/2-\kappa}\\
&\leq \tfrac {C_1}{2 t}\parb{\| \sqrt{W}\hat \chi(r/t)\e^{-\mathrm{i}tH}u\|^2+\| \sqrt{W} B\e^{-\mathrm{i}tH}u\|^2}+C_3t^{-1-\sigma}+C_2t^{-3/2-\kappa}.
\end{align*} It remains to apply  the bound \eqref{eq:45bk} with
$\chi$ as well as with $\chi\to \hat\chi$.
\end{proof}

The following type of result is called a {\it minimal  velocity bound} in the
literature.
\begin{lemma}
  \label{lemma:mini}
For  all  $ u\in \chi_{[E-e,E+e]}(H)\vH$ with $e>0$  taken
sufficiently small the bound \eqref{eq:12} holds.
\end{lemma}
\begin{proof} Let $\chi\in
  C^\infty(\R)$ be given such that $\chi'\in
  C^\infty_{\mathrm{c}}(\R_+)$  and the number
  $\sqrt{(1+\delta_2)^2E/2}$ is to the right of the support of $\chi$. Then
    we shall show that  
\begin{equation}
    \label{eq:45bkq}
    \int_1^\infty
\|\chi(r/t)\e^{-\mathrm{i}tH}u\|^2 t^{-1}\d
t<\infty,
  \end{equation} showing in particular \eqref{eq:12}.

Consider  the following uniformly bounded observables
 \begin{equation*}
  t^{-1}B^*
  AB;\;B:=\chi(r/t)\chi_e(H)\mand A \text { is given by } \eqref{eq:6}.
  \end{equation*} Due to Lemmas \ref{lem:comm-compss} and \ref{lemma:10.1.23.9.52}, for all
  sufficiently small $e>0$
\begin{equation*}
  t^{-1}B^*
  (\D_HA) B\geq  2(1+\delta_1)E t^{-1}B^* B+O(t^{-3}).
  \end{equation*}  Next, using this bound, Lemma \ref{lem:comm-compss} again and an
  estimation of the momentum in terms of the energy as in the proof of
  Lemma \ref{lemma:10.1.16.14.53}, we deduce
\begin{align*}
  B^*
   (\D_H\tfrac At) B&=t^{-1}B^*
   (\D_HA) B-t^{-2}B^*
  AB\geq  ct^{-1}B^*B+O(t^{-2});\\ c&=2\parbb{(1+\delta_1)-(1+\delta_2)\sqrt{1+2e/E}}E .
  \end{align*} Here $c>0$ for $e>0$ small enough.

To complete the  proof of the lemma  it suffices  to bound 
\begin{equation}
  \label{eq:50}
  \int_1^\infty
|\Re \inp{ (\D_{H_0} \chi(r/t))\e^{-\mathrm{i}tH}u,A\chi(r/t)\e^{-\mathrm{i}tH}u}|\,t^{-1}\d
t<\infty.
\end{equation} For that we also introduce $T:=t^{-1}A\chi(r/t)Q_3$ and use
\eqref{eq:18}, Lemma \ref{lem:comm-compss}  and notation of Lemma
\ref{lem:comm-comp} to write, with $\hat \chi \in C_{\mathrm{c}}^\infty(\R_+)$
chosen such that   $\hat \chi=1$ on the support of $ \chi $,
\begin{align*}
  &\inp{
    (\D_{H_0}\chi(r/t))\e^{-\mathrm{i}tH}u,A\chi(r/t)\e^{-\mathrm{i}tH}u}\\
&=\inp{
    (\Re{p_{\chi'(r/t)}}-r/t\chi'(r/t))\hat \chi(r/t)
    \e^{-\mathrm{i}tH}u,T\hat \chi(r/t)
    \e^{-\mathrm{i}tH}u} +O(t^{-1/2-\kappa}).
\end{align*} Although it is here legitimate to replace
$\Re{p_{\chi'}}$ by $p_{\chi'}$ and $A$ by
    $(\partial_ir^2)g^{ij}p_j$ it is preferable to keep the
    symmetrized form. First we note that the energy localization (implemented by the
appearance of the factor $Q_3$) makes $\|T\|$
uniformly bounded. We claim that for any $\sigma\in (0,1+2\kappa]$  also the operators
\begin{equation*}
  S:=(W+t^{-\sigma})^{-1/4}\Re\parb{T^*(\Re{p_{\chi'}}-r/t\chi')}(W+t^{-\sigma})^{-1/4}
\end{equation*} have uniformly bounded norm. Given this property 
 we can bound the integral
\eqref{eq:50} by
\begin{equation*}
  \int_1^\infty
|\inp{
    (W+t^{-\sigma})^{1/4}\hat\chi(r/t)\e^{-\mathrm{i}tH}u,S(W+t^{-\sigma})^{1/4}\hat \chi(r/t)\e^{-\mathrm{i}tH}u}|\,t^{-1}\d
t+C,
\end{equation*} and we conclude by invoking Lemma
\ref{lemma:key_prel-local-pert}.
To bound $S$ we note that interpolation yields
\begin{equation*}
  \|S\|\leq \|\Re\parb{T^*(\Re{p_{\chi'}}-r/t\chi')}(W+t^{-\sigma})^{-1/2}\|,
\end{equation*} and due to Lemma \ref{lem:comm-compssC} we can write 
\begin{equation*}
  \Re\parb{T^*(\Re p_{\chi'}-r/t\chi')}=T^*\chi'(r/t)(p_r-r/t)+O(t^{-1/2-\kappa}).
\end{equation*} Whence, using  here also \eqref{eq:15}, it  follows  that
\begin{equation*}
  \|S\|\leq \|T^*\|\times
  \|\chi'(r/t)(p_r-r/t)(W+t^{-\sigma})^{-1/2}\|+C_1t^{\sigma/2-1/2-\kappa}\leq C_2.
\end{equation*}
\end{proof}
\begin{corollary}\label{cor:10.6.9.4.24bb}
For  all  $ u\in \chi_{[E-e,E+e]}(H)\vH$ with $e>0$  taken
sufficiently small  \eqref{eq:41} holds.
\end{corollary}
\begin{proof}
  We use Lemma \ref{lemma:10.1.16.14.53}, Corollary
  \ref{cor:10.6.9.4.24}, \eqref{eq:45bkq} and the subsequence
  argument in Step 3 in the proof of Lemma
  \ref{lemma:10.1.16.14.53}. Indeed there exists a sequence $t_n\to \infty$ such
that
\begin{equation*}
\lim_{n\to \infty}q(t_n)=0;\;q(t):=\inp{\e^{-\mathrm{i}tH}u,(I-Q_3Q_4Q_5^2Q_4Q_3)\e^{-\mathrm{i}tH}u},
  \end{equation*} and the time-derivative of $q$ is integrable due to Corollary
  \ref{cor:10.6.9.4.24} and  \eqref{eq:45bkq}.
\end{proof}

\subsection{Proof of Lemma~\ref{lemma:10.3.23.17.20}}\label{sec:10.1.23.12.54}
We prove Lemma~\ref{lemma:10.3.23.17.20} alone the line of the proof
of Lemma \ref{lemma:10.5.31.23.24} using the properties
\eqref{eq:41}--\eqref{eq:12}.
So
 let $Q_{\mathrm{p}}$ be the operator defined by (\ref{10.6.6.2.52}),
 and let $e>0$ be small enough.
Then the property \ref{item:10.3.23.17.21} follows by mimicking the
proof of Lemma \ref{lemma:10.6.4.23.38} 
using \eqref{eq:41}--\eqref{eq:12}. This amounts to showing for $u\in
\chi_{[E-e,E+e]}(H){\mathcal H}$ and for the
same quantity $T_N=T_N(t)$ as before the
existence of  $R=R(t)\in{\mathcal B}({\mathcal
    H})$  such that 
\begin{subequations}
\begin{align}\label{eq:37k}
&\int_{t_0}^\infty |\inp{\e^{-\i tH}u,R\e^{-\i tH}u}|\,\d t =o(t_0^0)\text{ uniformly in } N\geq1,\\
&\tfrac{\mathrm{d}}{\mathrm{d}t}\inp{\e^{-\i tH}u,Q_4Q_5(I-T_N^{-1})Q_5Q_4\e^{-\i tH}u}
\le \inp{\e^{-\i tH}u,R\e^{-\i tH}u}.
\label{eq:10.6.4.13.13k}
\end{align}
  \end{subequations} We compute the derivative in
  \eqref{eq:10.6.4.13.13k}. The contribution from
  $\D_HT_N^{-1}=\D_{H_0}T_N^{-1}+\i[V,T_N^{-1}]$ is treated as before
  (note the trivial bound $Q_5\i[V,T_N^{-1}]Q_5=O(t^{-1-\eta})$). It
  remains to consider the contribution 
  \begin{equation*}
   2\Re \inp{\e^{\i tH}u,\parb {\D_HQ_4Q_5}(I-T_N^{-1}(t))Q_5Q_4\e^{-\i tH}u}. 
  \end{equation*} For that we compute
  $\D_HQ_4Q_5=(\D_HQ_4)Q_5+Q_4\D_HQ_5$, invoke \eqref{eq:18w} and
  Lemmas \ref{lem:comm-comp} and \ref{lem:comm-compss}, and use
  \eqref{eq:12d} and  \eqref{eq:12} to treat the contributions from
  $\D_HQ_4$ and $\D_HQ_5$, respectively. Note that here the implementation of
  \eqref{eq:12d} and  \eqref{eq:12}  requires symmetrization. For that
  part we use also Corollary \ref{cor:furth-comm-comp}.

As for the property
\ref{item:10.3.23.17.22} we use again use the proof of Lemma
\ref{lemma:10.6.4.23.38}. The contribution from $\D_HT_{N(t)}^{-1}
=\D_{H_0}T_{N(t)}^{-1}+\i[V,T_{N(t)}^{-1}]$
does not need elaboration. As for the contribution from $\D_HQ_4Q_5$
we compute as above using \eqref{eq:18w}. We claim that this
contribution indeed is $O(t^{-1-\delta'})$ for $\delta'\leq
(1+\delta_1)/2$, as may be  seen by using the
localization provided by the factor $Q_6^2$. Indeed due to Lemma
\ref{lem:comm-comp} we can bound for any $\chi\in C_{\rm c}^\infty (\R_+)$
\begin{equation*}
  \|\parb{\chi(r/t)(p_r-r/t)+\text {h.c.}} Q_6^2\| \leq C t^{-1/2-\delta_1/2}.
\end{equation*}

The property \ref{item:10.3.23.17.23} is proved by commuting as in the
proof of Lemma \ref{lemma:10.5.31.23.24} \ref{item:10.5.31.23.27}.
Thus the lemma is proved.
\endproof


\end{document}